%% file: arxiv.tex
\documentclass[11pt]{article}
\input{Inputs/header}

\title{Min-CSPs on Complete Instances}
\author{Aditya Anand\footnote{Supported in part by NSF grant CCF-2236669.} \and Euiwoong Lee\footnote{Supported in part by NSF grant CCF-2236669 and Google.} \and Amatya Sharma\footnote{Supported in part by NSF grant CCF-2236669.}}
\date{University of Michigan, Ann Arbor}

\begin{document}

    \maketitle
    \begin{abstract}
        \input{Inputs/abstract}
    \end{abstract}
    
    \thispagestyle{empty}
    \newpage
    \clearpage
    \pagenumbering{arabic} 
    
    \input{Inputs/intro}
    
    \input{Inputs/2sat}

    \input{Inputs/3csp}

    \input{Inputs/kcsp_vecv}

    \input{Inputs/induced2sat}

    \input{Inputs/23csp}

\input{Inputs/pac}

    \input{Inputs/hardness}

    \bibliographystyle{alpha}
    \bibliography{Inputs/references.bib}

\end{document}

%% file: Inputs/header.tex
\date{\today}

\usepackage[letterpaper]{geometry}
\usepackage[utf8]{inputenc}
\usepackage{amsthm}
\usepackage{amsmath}
\usepackage{amsfonts}
\usepackage{amssymb}
\usepackage[hidelinks]{hyperref}
\usepackage{cleveref}
\usepackage{xcolor}
\usepackage{fullpage}
\usepackage{algorithm}
\usepackage{algpseudocode}
\usepackage{xspace}
\usepackage{multirow}
\usepackage{diagbox}
\usepackage[colorinlistoftodos,textsize=tiny,textwidth=2cm,color=green!50!gray]{todonotes}

\usepackage{thmtools, thm-restate} 

\usepackage{bm}

\usepackage{silence}

\algnewcommand{\IIf}[1]{\State\algorithmicif\ #1\ \algorithmicthen}
\algnewcommand{\EndIIf}{\unskip\ \algorithmicend\ \algorithmicif}

\usepackage{derivative} 
\usepackage{mathtools} 
\usepackage{bbm} 
\usepackage{tikz}
\usepackage{forest}
\newtheorem{theorem}{Theorem}[section]
\newtheorem{corollary}[theorem]{Corollary}
\newtheorem{claim}[theorem]{Claim}
\newtheorem{lemma}[theorem]{Lemma}

\newtheorem{definition}[theorem]{Definition}
\newtheorem{remark}[theorem]{Remark}

\newtheorem{observation}[theorem]{Observation}

\ifdefined\DEBUG

    \newcommand{\aay}[1]{\textcolor{blue}{#1}}
    
    \def\rem#1{{\marginpar{\raggedright\scriptsize #1}}}
    \newcommand{\adir}[1]{\rem{\textcolor{red}{$\bullet$ #1}}}
    \newcommand{\euir}[1]{\rem{\textcolor{green}{$\bullet$ #1}}}
    \newcommand{\aayr}[1]{\rem{\textcolor{blue}{$\bullet$ #1}}}

    \newcommand{\amatya}[1]{\todo[color=blue!100!black!50]{A: #1}}
    \newcommand{\euiwoong}[1]{\todo[color=green!100!black!50]{E: #1}}
\else

    \newcommand{\aay}[1]{#1}
    \newcommand{\adir}[1]{}
    \newcommand{\euir}[1]{}
    \newcommand{\aayr}[1]{}
    \newcommand{\euiwoong}[1]{}
    \newcommand{\amatya}[1]{}
\fi

\newcommand{\sat} {\ensuremath{\mathsf{sat}}\xspace}
\newcommand{\unsat} {\mathsf{unsat}}

\newcommand{\maxcsps} {\textsc{Max-CSP}s\xspace}
\newcommand{\maxcsp} {\textsc{Max-CSP}\xspace}
\newcommand{\ug} {\textsc{Unique Games}\xspace}
\newcommand{\mincsp}{\textsc{Min-CSP}\xspace}
\newcommand{\minkcsp}{\textsc{Min-$k$-CSP}\xspace}
\newcommand{\minksat}{\textsc{Min-$k$-SAT}\xspace}
\newcommand{\kcsp}{$k$-\textsc{CSP}\xspace}
\newcommand{\kcsps}{$k$-\textsc{CSPs}\xspace}
\newcommand{\klin}{$k$-\textsc{LIN}\xspace}
\newcommand{\kand}{$k$-\textsc{AND}\xspace}
\newcommand{\krcsp}[2]{\textsc{$(#1,#2)$-CSP}\xspace}
\newcommand{\ncsp}[1]{\textsc{$#1$-CSP}\xspace}
\newcommand{\mincsps}{\textsc{Min-CSP}s\xspace}

\newcommand{\minuncut}{\textsc{Min-Uncut}\xspace}
\newcommand{\mintwo}{\textsc{Min-2-SAT}\xspace}
\newcommand{\maxtwo}{\textsc{Max-2-SAT}\xspace}
\newcommand{\cc} {\textsc{Correlation Clustering}\xspace}
\newcommand{\lra}{\textsc{Low Rank Approximation}\xspace}

\newcommand{\maxcut}{\ensuremath{\textsc{Max-Cut}}}

\newcommand{\nae}{\ensuremath{\textsc{NAE-$3$-SAT}}}

\newcommand{\nsat}[1]{\ensuremath{\textsc{$#1$-SAT}}\xspace}

\newcommand{\ksat}{\nsat{k}}

\newcommand{\threesat}{\nsat{3}}
\newcommand{\twosat}{\nsat{2}}

\newcommand{\csp}{\ensuremath{\textsc{CSP}}\xspace}

\newcommand{\pac}[1]{\ensuremath{\textsc{$(#1)$\mbox{-}PAC}}\xspace}

\newcommand{\itwocsp}[1]{\ensuremath{\textsc{$#1$\mbox{-}Induced\ncsp{2}}}\xspace}
\newcommand{\kitwocsp}{\itwocsp{k}}

\newcommand{\poly}{\mbox{\rm poly}}
\newcommand{\polylog}{\mbox{\rm  polylog}}

\newcommand{\sdp}{\ensuremath{\mathsf{sdp}}}
\newcommand{\eps}{\varepsilon}

\newcommand{\argmax}{\mathrm{argmax}}

\newcommand{\veca}{\ensuremath{\mathbf{a}}}

\newcommand{\vecell}{\ensuremath{\bm{\ell}}}

\newcommand{\vecv}{\ensuremath{\mathbf{v}}}

\newcommand{\cA}{\ensuremath{\mathcal{A}}}

\newcommand{\cC}{\ensuremath{\mathcal{C}}}
\newcommand{\calC}{\ensuremath{\mathcal{C}}}

\newcommand{\calD}{\ensuremath{\mathcal{D}}}

\newcommand{\R}{\ensuremath{\mathbb{R}}}

\newcommand{\cS}{\ensuremath{\mathcal{S}}}

%% file: Inputs/abstract.tex
Given a fixed arity $k \geq 2$, \minkcsp on complete instances is the problem whose input consists of a set of $n$ variables $V$ and one (nontrivial) constraint for every $k$-subset of variables (so there are $\binom{n}{k}$ constraints), and the goal is to find an assignment that minimizes the number of unsatisfied constraints. 
Unlike \textsc{Max-$k$-CSP} that admits a PTAS on more general dense or expanding instances, the approximability of \minkcsp has not been well understood. Moreover, for some CSPs including \minksat, there is an approximation-preserving reduction from general instances to dense and expanding instances, leaving complete instances as a unique family that may admit new algorithmic techniques. 

In this paper, we initiate the systematic study of \mincsps on complete instances. 
First, we present an $O(1)$-approximation algorithm for \mintwo on complete instances, the minimization version of \maxtwo. Since $O(1)$-approximation on dense or expanding instances refutes the Unique Games Conjecture, it shows a strict separation between complete and dense/expanding instances. 

Then we study the decision versions of CSPs,  
whose goal is to find an assignment that satisfies all constraints; an algorithm for the decision version is necessary for any nontrivial approximation for the minimization objective. Our second main result is a quasi-polynomial time algorithm for every Boolean $k$-CSP on complete instances, including $\ksat$. 
We complement this result by giving additional algorithmic and hardness results for CSPs with a larger alphabet, yielding a characterization of (arity, alphabet size) pairs that admit a quasi-polynomial time algorithm on complete instances.

%% file: Inputs/intro.tex
\section{Introduction}

Constraint Satisfaction Problems (CSPs) form arguably the most important class of computational problems.
Starting from the first NP-complete problem of \threesat, they also have played a crucial role in approximate optimization, resulting in (conditional) optimal approximabilities of fundamental problems including 
 \textsc{Max}-3\textsc{LIN}, \textsc{Max}-3\textsc{SAT}, \textsc{Max}-\textsc{Cut}, and \ug~\cite{Hastad01, Khot02, KKMO07}.
The study of CSPs has produced numerous techniques that are widely used for other optimization problems as well; famous examples include linear program (LP) and semidefinite program (SDP) rounding from the algorithms side, and
probabilistically checkable proofs (PCPs) and the Unique Games Conjecture (UGC) from the hardness side.

While the aforementioned results concern {\em general instances}, there also have been active studies on CSPs on {\em structured instances} where the input is promised to exhibit a certain structure.
For $k$-CSPs where each constraint depends on $k$ variables and the structure of the constraints can be represented as a $k$-uniform hypergraph $H = (V, E)$ with $n = |V|$ vertices, 
two of the most studied structured instances are 
{\em dense instances}, where we have $|E| = \Omega(n^k)$ and 
{\em (high-dimensional) expanding instances}, where $H$ exhibits some expansion property 
(which can be further generalized to low-threshold rank graphs and (certifiable) small-set expanders).

For \maxcsps whose goal is to maximize the number of satisfied constraints, the research on dense/expanding instances has produced beautiful algorithmic techniques, most notably based on {\em random sampling}~\cite{arora1995polynomial, 
bazgan2003polynomial, alon2003random, de2005tensor, mathieu2008yet, KS09, barak2011subsampling, yaroslavtsev2014going, manurangsi2015approximating, fotakis2016sub}, {\em convex hierarchies}~\cite{de2007linear, arora2008unique, BRS11, GS11, yoshida2014approximation, alev2019approximating, jeronimo2020unique, bafna2021playing}, and {\em regularity lemmas}~\cite{frieze1996regularity, coja2010efficient, oveis2013new, jeronimo2021near}.
The conclusion is: {\em every} \maxcsp admits a PTAS on dense/expanding instances via {\em any} of these tools. 
{\em Complete instances}, where $E = \binom{V}{k}$, is obviously a special case of both dense and expanding instances, so the existence of a PTAS is an immediate corollary for \maxcsp. 

In this paper, we initiate the systematic study of \mincsps on complete instances, where the goal is to minimize the number of constraints not satisfied by the output assignment. As an example consider the problem of \mintwo (also known as \mintwo \textsc{Deletion}). The input for this problem is a complete instance on variables $V$ and \twosat clauses $\cC$ such that for every pair of variables $v_1,v_2\in V$, there is a clause $(\ell_1\vee \ell_2)\in \cC$, for some literals $\ell_i \in \{v_i, \neg v_i\}$, $\forall i\in[2]$. The goal is to delete minimum number of clauses so that the remaining instance is satisfiable.

While the \mincsps are equivalent to their maximization counterparts in terms of exact optimization, minimization versions often exhibit more interesting characterizations of their approximabilities and demand a deeper understanding of algorithms and hardness techniques;
for general instances, while every \maxcsp (with a constant alphabet size) trivially admits an $\Omega(1)$-approximation, \mincsps with the Boolean alphabet already exhibit an interesting structural characterization that identifies the optimal approximability of every \mincsp as one of $\{ 1, O(1), \polylog(n), \poly(n), \infty \}$-approximations~\cite{khanna2001approximability}.

Compared to the almost complete literature on \maxcsps on general/structured instances and \mincsps on general instances, the study of \mincsps on structured instances is not as rich. Furthermore, we believe that a thorough investigation of structured \mincsps will (1) obtain a fine-grained understanding of instance structures, (2) establish strong connections between CSPs and problems arising in data science and machine learning, and (3) unify various algorithmic techniques and yield new ones. They are briefly elaborated in the following.

\paragraph{Fine-grained understanding of structures.}
Compared to \maxcsps where every problem admits a PTAS on the previously discussed instances, \mincsps will exhibit more interesting characterizations of their approximabilities depending on specific predicates and structures.
For instance, for \minuncut (the minimization version of \maxcut), extensions of the techniques used for \maxcsps yields an $O(1)$-approximation~\cite{GS11, KS09} on expanders and dense graphs.
However, for \mintwo (the minimization version of \maxtwo)\footnote{This problem is also known as \textsc{Min}-2\textsc{CNF Deletion}. Throughout the paper, we use the same name for both max and min versions, the only exception being $\minuncut$ and \textsc{Min}-\textsc{Undicut} (the min version of \textsc{Max}-\textsc{Dicut}).}, one can easily encode a general hard instance in dense graphs and expanders (see Claim~\ref{claim:hardness-dense}), implying that the problem in dense/expanding instances will be unlikely to admit an $O(1)$-approximation. 
Our first main result (Theorem~\ref{thm:2sat}) shows an $O(1)$-approximation for \mintwo on complete instances, formally separating complete and dense/expanding instances. 


\paragraph{Connections to data science and machine learning.}
It turns out that numerous optimization problems arising in data science, mostly notably on clustering, metric fitting, and low-rank approximation, can be cast as a \mincsp on complete (bipartite) instances; indeed, most known algorithms for these problems use techniques initially developed for CSPs. 

Well-known examples include \cc~\cite{CMSY15, CLN22, CLLN23, CCLLNV24} and \lra~\cite{ban2019ptas, cohen2023ptas}.\footnote{The connections from CSPs to \cc and \lra are elaborated in Section 1 of~\cite{CLN22} and Section 1.1 of~\cite{cohen2023ptas} respectively.}
Recently, minizing the Kamada-Kawai objective for Multidimensional Scaling has been studied from the approximation algorithms perspective, using hierarchy-based tools for CSPs~\cite{demaine2021multidimensional, bakshi2023quasi}.

While these connections between CSPs and the aforementioned problems in DS/ML (and more) remain at problem-specific levels, we believe that the comprehensive study of \mincsps will strengthen these connections and yield new insights.

\paragraph{Unifying and creating techniques.}
As mentioned before, PTASes for Max CSPs on dense or expanding instances can be achieved by three different techniques: random sampling, regularity lemmas, and convex hierarchies.
They are independently interesting methods that have applications in other areas, and it would be desirable to explore connections between them and have a unified understanding.
\maxcsps are not the perfect class of problems to study their relative powers and connections, since {\em any technique can do everything}.

The study of \mincsps and related problems in data science will provide a more interesting arena to do so. For many problems in this new space, their best-known approximation is achieved by only one method out of the three, and recovering the same guarantee by the other two methods seems to require a fundamentally better understanding of the techniques.
For example, a PTAS for \minuncut on dense graphs is currently only achieved by random sampling~\cite{KS09} while the best algorithm for \cc uses convex hierarchies~\cite{CCLLNV24}. Is this separation between techniques inherent, or will they eventually become equivalent again?


\subsection{Our Results}
\label{sec:results}
As previously mentioned, any \maxcsp on dense/expanding graphs admits a PTAS.
Such success for \maxcsps was partially transferred to certain \mincsps, most notably $O(1)$-approximation algorithms (sometimes PTASes) for \ug and \minuncut~\cite{bazgan2003polynomial, KS09, GS11, meot2023voting}. 
More generally, Karpinski and Schudy~\cite{KS09} presented a PTAS for every {\em fragile} CSP~\cite{KS09} on dense instances, where a CSP is fragile when changing the value of a variable always flips a clause containing it from satisfied to unsatisfied. 
If we restrict our attention to Boolean (each variable has two possible values) binary (each constraint contains two variables) CSPs, \minuncut, \textsc{Min}-\textsc{Undicut}, and \textsc{Min}-2-\textsc{AND} are fragile, essentially leaving only \mintwo. 

However, 
the following simple reduction shows that \mintwo on dense (and expanding) instances is as hard as general instances: given a general instance for \mintwo with $n$ variables and $m$ constraints, add $n' = \Omega(n/\eps)$ dummy variables, and for every pair $(u, v)$ of variables that contains at least one dummy, add a constraint $(u \vee v)$. Then the new instance has $\binom{n' + n}{2} - \binom{n}{2} + m \geq (1 - \eps)\binom{n' + n}{2}$ constraints, and every assignment of the original instance has the same value as the corresponding assignment of the new instance that gives True to every dummy variable, which implies that the optimal values of these two instances are the same. 
Since \mintwo on general instances does not admit an $O(1)$-approximation algorithm assuming the UGC~\cite{Khot02}, the same hardness holds for dense and expanding instances.

Our first main result is an $O(1)$-approximation algorithm for \mintwo, which separates complete instances from dense/expanding instances assuming the UGC.

\begin{restatable}{theorem}{2sat}
There is a polynomial-time $O(1)$-approximation algorithm \textup{\mintwo} on complete graphs. 
\label{thm:2sat}
\end{restatable}

What about other Boolean CSPs? In order to have any nontrivial approximation algorithm for some \mincsp, it is necessary to have an algorithm for its {\em decision version}, which decides whether there is an assignment satisfying all the constraints; otherwise, it is impossible to decide whether the optimal value is zero or not. While \mintwo had classical algorithms for its decision version \twosat even on general instances, there are numerous CSPs whose decision version is NP-hard on general instances, most notably \ksat.\footnote{Throughout the paper, \kcsp or \ksat without the prefix \textsc{Max} or \textsc{Min} denotes the decision version.}  

To formally state the result, let $\Sigma = \{0, 1 \}$ be the {\em alphabet} of Boolean CSPs.
Let \kcsp be the computational problem whose input consists of the set of variables $V = \{ v_1, \dots, v_n \}$ and the set of constraints $\calC$ where each $C \in \calC$ is a subset of size $k$ (called $k$-set onwards) variables and comes with a predicate $P_C : \Sigma^C \to \{ \sat, \unsat \}$.\footnote{We will assume that each $k$-set has at most one constraint unless mentioned otherwise.}
The goal is to find an assignment $\alpha : V \to \Sigma$ such that $P_C(\alpha(C)) = \sat$ for every $C \in \mathcal{C}$, where 
$\alpha(C) := (\alpha(v_i))_{v_i \in C}$ be the restriction of $\alpha$ to $C$.
An instance $(V, \calC)$ of \kcsp is {\em complete} if every $k$-subset of $V$ corresponds to exactly one constraint $C \in \calC$ (so $|\calC| = \binom{n}{k})$), and for every $P_C$, there exists at least one $\sigma \in \Sigma^k$ such that $P(\sigma) = \unsat$. (Without the second restriction, every instance of general \kcsp can be reduced to a complete instance by putting the trivial predicate that accepts every string to each $k$-set that did not have a constraint.) This framework captures well-known CSPs including \kcsp, \kand, \klin on complete instances. 

Note that the above reduction from general to dense instances for \twosat immediately extends to \ksat, implying that \ksat and more generally \kcsp are hard on dense instances (see Claim~\ref{claim:hardness-dense}).
Our second main result presents a quasi-polynomial time algorithm for complete \kcsp.

\aayr{should we mention the all satisfying assignments part here?}
\begin{restatable}{theorem}{kcspassgn}
For any constant $k$,
there is an $n^{O(\log n)}$-time algorithm that decides whether a given complete instance for \textup{\kcsp} is satisfiable or not. 

\label{thm:kcspassgn}
\end{restatable}



What happens beyond the Boolean alphabet? For $r \geq 3$, let \krcsp{k}{r} be the \kcsp defined above with the only difference being the alphabet $\Sigma = \{0, 1, \dots, r - 1 \}$. The following theorems give a complete classification of \krcsp{k}{r} that admits a quasi-polynomial time algorithm (assuming $\mathbf{NP} \subsetneq \mathbf{QP}$); the only tractable cases are $(k, 2)$ for any $k$ and $(2, 3)$.

\begin{restatable}{theorem}{twothreecsp}
There is an $n^{O(\log n)}$-time algorithm that decides whether a given complete instance for \textup{\krcsp{2}{3}} is satisfiable or not. 
\label{thm:23csp}
\end{restatable}


\begin{theorem}
It is NP-hard to decide whether a given complete instance for \textup{\krcsp{2}{4}} is satisfiable or not. 
\label{thm:24csp}
\end{theorem}

\begin{restatable}{theorem}{threethreecsp}
It is NP-hard to decide whether a given complete instance for \textup{\krcsp{3}{3}} is satisfiable or not. 
\label{thm:threethreecsp}
\end{restatable}



\paragraph{Open questions.}
There are many open questions following our results. Most importantly, the approximability of Boolean \mincsps on complete instances is wide open; it is not ruled out that \minkcsp for any constant $k$ admits a PTAS on complete instances, while we do not even have it for \mintwo! (We show in Theorem~\ref{thm:hardness-min} that the exact optimization is hard unless $\mathbf{NP} \subseteq \mathbf{BPP}$ for many \mincsps). 

It is also an interesting question to see whether our $n^{O(\log n)}$-time algorithms can be made polynomial-time or not. It is notable that $n^{\Theta(\log n)}$ is the optimal running time for having a constant approximation for \maxcsp with $k = 2$ and $r = \poly(n)$ on complete bipartite instances assuming the Exponential Time Hypothesis~\cite{aaronson2014multiple, manurangsi2016birthday}. 

Finally, a lot of research on CSPs has focused on the effect of the {\em constraint language} on computational tractability.
In our context, given $k$ and $\Sigma$, the constraint language $\Gamma = \{ P_1, \dots P_t \}$ is defined to be a finite set of predicates where each $P_i : \Sigma^k \to \{ \sat, \unsat \}$ has at least one input leading to $\unsat$. A complete instance of CSP($\Gamma$), for every $C \in \binom{V}{k}$, has a constraint with $P_C \in \Gamma$. This framework still captures well-known CSPs including \ksat, \kand, \klin. 
While our positive result Theorem~\ref{thm:kcspassgn} works for every such CSP with Boolean alphabet, our negative results like
Theorem~\ref{thm:24csp} and 
Theorem~\ref{thm:threethreecsp} still leave the possibility that many interesting constraint languages are indeed tractable on complete instances, calling for a more fine-grained characterization depending on $\Gamma$. 


\paragraph{Organization and additional results.}
Section~\ref{sec:2sat} presents an $O(1)$-approximation algorithm for \mintwo, proving Theorem~\ref{thm:2sat}. 
A quasi-polynomial time algorithm for \threesat and eventually \kcsp will be presented in Section~\ref{sec:3sat} and \ref{sec:kcsp} respectively. Section~\ref{sec:nae} contains polynomial-time algorithms for some complete CSPs that are NP-hard on general instances, including \nae.

Section~\ref{sec:23csp} discusses binary, non-Boolean \krcsp{2}{r} with $r \geq 3$ and proves Theorem~\ref{thm:23csp}. \Cref{sec:pac} presents additional algorithmic results for $5$-\textsc{Permutation Avoid Coloring }($5$-\textsc{PAC}), which is a generalization of $5$-\textsc{Coloring} where each edge can rule out any permutation between colorings of the two vertices (instead of just the identity permutation). 

Our hardness results, which include Theorems~\ref{thm:24csp} and~\ref{thm:threethreecsp}, are proved in Section~\ref{sec:hardness}. It also proves the NP-hardness of $6$-\textsc{PAC}, showing that five is the limit for \textsc{PAC}. 

%% file: Inputs/2sat.tex
\section{\mintwo}
\label{sec:2sat}
In this section, we present an $O(1)$-approximation for \mintwo on complete instances, proving Theorem~\ref{thm:2sat}.
Let $\Sigma = \{ 0, 1 \}$ be the Boolean alphabet, where $1$ indicates True and $0$ denotes False. 
Let $X$ be the set of $n$ variables that can take a value from $\Sigma$.

For Boolean \ksat, we slightly simplify the notation and represent each clause $C$ as an unordered $k$-tuple of literals $(\ell_1, \dots, \ell_k)$, which denotes the \ksat clause $(\ell_1 \vee \dots \vee \ell_k)$.
So for $k = 2$, and for every pair of variables $( x, y ) \in \binom{X}{2}$, there exists exactly one constraint $C = (\ell_x, \ell_y)$ where $\ell_x \in \{ x, \neg x \}$ and $\ell_y \in \{ y, \neg y \}$. 
Let $L = X \cup \neg X$ be the set of $2n$ literals. (More generally, given $S \subseteq L$, let $\neg S := \{ \neg \ell : \ell \in S \}$.)

\paragraph{Relaxation.}
We use the standard SDP (vector) relaxation for \mintwo used in the $O(\sqrt{\log n})$-approximation on general instances~\cite{ACMM05}, which can be captured by a degree-4 Sum-of-Squares relaxation. The variables are $\{ v_{\ell} \in \R^{2n+1} : \ell \in L \} \cup \{ v_0 \}$ where $v_0$ is the special vector corresponding to $+1$ (True). All the vectors will be unit vectors. 
For two vectors $u, v \in \R^{2n+1}$, $d(u, v) := (1 + \langle v_0, u \rangle - \langle v_0, v \rangle - \langle u, v \rangle) / 4$ denotes the directed distance from $u$ to $v$. For instance, when $u, v = \pm v_0$, among four possibilities, it is $1$ if $u = v_0, v = -v_0$ and $0$ in the other cases. Also note that $d(u, v) = 
d(-v, -u)$.

\begin{align*}
\mbox{minimize } \quad & \sum_{(\ell, \ell') \in \calC} d(v_{\neg \ell}, v_{\ell'}) 
\\ 
 & \| v_0 \|_2^2 = 1 \\
 & \| v_{\ell} \|_2^2 = 1 
 && \mbox{ for every } \ell \in L \\
 & v_{\ell} = -v_{\neg \ell} \quad \Leftrightarrow\quad  \| v_{\ell} + v_{\neg \ell} \|_2^2 = 0 && \mbox{ for every } \ell \in L \\
 & \| v_{\ell_1} - v_{\ell_2} \|_2^2 + 
 \| v_{\ell_2} - v_{\ell_3} \|_2^2 
 \geq 
 \| v_{\ell_1} - v_{\ell_3} \|_2^2
 && \mbox{ for every } \ell_1, \ell_2, \ell_3 \in L \cup \{ 0 \}. 
\end{align*}
This is a valid relaxation because for given any assignment $\alpha : X \to \Sigma$, the SDP solution where we let $v_x = v_0$ if $\alpha(x) = 1$ and $v_x = -v_0$ otherwise yields the same objective function value as the number of 2SAT constraints unsatisfied by $\alpha$.

Given the optimal SDP solution $\{ v_{\ell} \}_{\ell}$, let $b_{\ell} := \langle v_0, v_{\ell} \rangle \in [-1, +1]$ indicate the {\em bias} of $\ell$ towards True. Given two literals $p, q \in L$, let $d(p, q) := d(v_p, v_q)$. Another interpretation of $d(p, q)$ is \[
d(p, q) = 
\frac{1 + \langle v_0, v_p \rangle - \langle v_0, v_q \rangle - \langle v_p, v_q \rangle}{4}
= \frac{\| v_p - v_q \|_2^2}{8} + \frac{b_p - b_q}{4},
\]
which is the usual $\ell_2^2$ distance between $v_p$ and $v_q$, adjusted by the difference between the biases. So it is easy to see the directed triangle inequality $d(p, q) + d(q, r) \geq d(p, r)$ for any $p, q, r \in L \cup \{0 \}$. 

\paragraph{Formulation as graph cuts.}
We also use the well-known interpretation of \mintwo as a special case of \textsc{Symmetric Directed Multicut}, which was used in the previous $\polylog(n)$-approximation algorithms for \mintwo on general instances~\cite{klein1997approximation, ACMM05}. Let us consider the directed {\em implication} graph $G = (L, E)$ where the $2n$ literals become the vertices and each clause $(\ell \vee \ell')$ creates two directed edges $(\neg \ell, \ell')$ and $(\neg \ell', \ell)$; it means that if $\neg \ell$ (resp. $\neg \ell'$) becomes true, in order to satisfy the clause, $\ell'$ (resp. $\ell$) must become true. It is a classical fact that the 2SAT instance is satisfied if and only if the implication graph satisfies the following {\em consistency property}: no pair $\ell$ and $\neg \ell$ are in the same strongly connected component. 
Therefore, \mintwo is equivalent to finding the smallest set of clauses such that deleting the edges created by them yields the consistency property. 
Up to a loss of factor of $2$, we can consider finding the smallest set of edges whose deletion yields the consistency property.\footnote{Our algorithm will ensure that we delete $(\neg \ell, \ell')$ if and only if we delete $(\neg \ell', \ell)$, so we do not lose this factor $2$.} 

\paragraph{Symmetry.}
For a vector $v$ and $\eps \geq 0$, let $B_{out}(v, \eps) := \{ \ell \in L : d(v, v_{\ell}) \leq \eps \}$ and 
$B_{in}(v, \eps) := \{ \ell \in L : d(v_{\ell}, v) \leq \eps \}$; $B_{out}(v, \eps)$ (resp. $B_{in}(v, \eps)$) is the set of literals whose vector has a small directed distance {\em from} (resp. {\em to}) $v$. 
Also, in the implication graph $G = (L, E)$ and a subset $S \subseteq L$, let $\partial^+(S)$ be the set of outgoing edges from $S$ and $\partial^-(S)$ be the set of incoming edges into $S$. 
The symmetric constraint $v_{\ell} = -v_{\neg \ell}$, together with the symmetry in the construction of $G$ yields the following nice properties.

\begin{lemma}
For any vector $v$, number $\eps \geq 0$, and $\ell, \ell' \in L$, we have $(\ell, \ell') \in \partial^+(B_{out}(v, \eps))$ if and only if 
$(\neg \ell', \neg \ell) \in \partial^-(B_{in}(-v, \eps))$. In particular, 
$B_{out}(v, \eps) = \neg B_{in}(-v, \eps)$. 
\end{lemma}
\begin{proof}
Suppose $(\ell, \ell') \in \partial^+(B_{out}(v, \eps))$, which implies that (1) $d(v, v_{\ell}) \leq \eps$, (2) $d(v, v_{\ell'}) > \eps$, and (3) $(\ell, \ell') \in E$.

The definition of $d(\cdot, \cdot)$ ensures $d(u, u') = d(-u', -u)$, so we have (1) $d(v_{\neg \ell}, -v) \leq \eps$ and (2) $d(v_{\neg \ell'}, -v) > \eps$. 
The construction of $G$ ensures (3) $(\neg \ell', \neg \ell) \in E$, so we have $(\neg \ell', \neg \ell) \in \partial^-(B_{in}(-v, \eps))$. The argument in the other direction is the same. 
\end{proof}

\paragraph{Algorithm.}
Our relatively simple Algorithm~\ref{alg:2sat} is presented below. After the preprocessing step that removes literals with large absolute bias, we perform the well-known {\em CKR rounding} due to Carlinescu, Karloff, and Rabani~\cite{calinescu2005approximation} adapted to directed metrics. 

\begin{algorithm}[!htb]
        \caption{$O(1)$-Approximation for \mintwo}
        \label{alg:2sat}
        \begin{algorithmic}[1]
            \State Draw a random $\theta \in [0.001, 0.002]$. 
            \Comment{Preprocessing begin}
            \State Let $B_+ := \{ \ell \in L: b_{\ell} \geq \theta \}$ and 
            $B_- := \neg B_+ = \{ \ell \in L: b_{\ell} \leq -\theta \}$. \State Delete all outgoing edges from $B_+$ and incoming edges into $B_-$. \label{alg2satstep:preprocessing}
            \State Remove $B_+$ and $B_-$ from $L$ and corresponding variables from $X$. 
            \Comment{Preprocessing end}
            \State 
            \State Randomly order $X$.\Comment{CKR rounding begin}\label{alg2satstep:random}
            \For{each $x \in X$ in this order}
                \State Let $\ell \in \{ x, \neg x \}$ be the literal such that $b_{\ell} \geq 0$. (If $b_{x} = b_{\neg x} = 0$, let $\ell \leftarrow x$.) \label{alg2satstep:upcenter}

                \State Sample $\gamma \in [0.1, 0.11]$. 

                \State Delete all outgoing edges from $B_{out}(v_{\ell}, \gamma)$.            \label{alg2satstep:outball}

                \State Delete all incoming edges into $B_{in}(v_{\neg \ell}, \gamma)$.\label{alg2satstep:inball}

                \State Remove $B_{out}(v_{\ell}, \gamma) \cup B_{in}(v_{\neg \ell}, \gamma)$ from $L$ and corresponding variables from $X$. 
            \EndFor \Comment{CKR rounding end}
        \end{algorithmic}
    \end{algorithm}


Here the word {\em deletion} applies to the edges of $G$ that we pay in the objective function, and the word {\em removal} applies to the edges or vertices of $G$ that are removed from the current instance. So the deletion of an edge $(p, q)$ implies its removal, but necessarily not vice versa. 

We first prove the feasibility of the algorithm by proving that the implication graph after all the deletions satisfies the consistency property as follows:

\begin{lemma}
    The implication graph $G$ after deleting the edges deleted during the \Cref{alg:2sat}, satisfies the consistency property.
\end{lemma}
\begin{proof}
        At a high level, the algorithm picks a set $B$ of vertices (literals), deletes all the outgoing edges from $B$ or all the incoming edges into $B$, and removes $B$ from the instance. It is clear from the definition that neither $B_+$ nor $B_-$ can contain $p$ and $\neg p$ simultaneously for any literal $p \in L$. 
        In order to show that the same is true for  $B_{out}(v_{\ell}, \gamma)$ with $b_{\ell} \geq 0$ and $\gamma \leq 0.11$, note that 
        \[
        d(v_\ell, v_p) + 
        d(v_\ell, -v_p)
        = \frac{\| v_\ell - v_p \|_2^2 + \| v_\ell - (-v_p) \|_2^2}{8} + \frac{2b_\ell}{4}
        \geq \frac{\| v_\ell - v_p \|_2^2 + \| v_\ell + v_p \|_2^2}{8} 
        = \frac{1}{2},
        \]
        so it cannot be the case that both $d(v_\ell, v_p)$ and $d(v_\ell, v_{\neg p})$ are at most $\gamma \leq 0.11$. The same argument works for $B_{in}(v_{\neg \ell}, \gamma)$.
        Therefore, the original implication graph, after deleting the edges deleted during the algorithm, will have the consistency property. 
\end{proof}

\subsection{Analysis}
We analyze that the expected number of deleted edges is $O(\sdp)$, where $\sdp$ is the optimal value of the SDP 
relaxation. 

\paragraph{Case 1: Preprocessing.}
Let us bound the number of deleted edges in the preprocessing (Step~\ref{alg2satstep:preprocessing} of Algorithm~\ref{alg:2sat}). 
Fix an edge $(p, q)$. As $\theta \in [0.001, 0.002]$ is drawn from the interval of length $0.001$, the probability that $(p, q)$ belongs to $\partial^+(B_+)$ is at most $1000(b_p - b_q)^+$ (the same holds for $\partial^-(B_-)$), while it is contributing $d(p, q) \geq (b_p - b_q)^+/4$ to the SDP. (Let $a^+ := \max(a, 0)$.) Therefore, the expected total number of edges deleted in the preprocessing is at most $8000\sdp$. 

\paragraph{Case 2: Big edges.} We argue that we can only focus on the edges with small SDP contributions. 
Let $L$ and $X$ be the set of literals and variables after the preprocessing and $n' := |X|$.
Then $b_{\ell} \in [-0.002, 0.002]$ for every vertex $\ell \in L$. 
Let us consider a pair of variables $\{ x, y \} \in \binom{X}{2}$. Let $\ell_x \in \{ x, \neg x \}$ and $\ell_y \in \{ y, \neg y \}$ such that $(\ell_x, \ell_y) \in E$. (I.e., $(\neg \ell_x \vee \ell_y)$ is the clause for $\{ x, y \}$.)
Let $\eps_0 > 0$ be the constant to be determined. If $d(\ell_x, \ell_y) > \eps_0$, just deleting the edge is locally a $(1/\eps_0)$-approximation, so we can focus on edges with $d(\ell_x, \ell_y) \leq \eps_0$.



\paragraph{Case 3: Improved CKR analysis.}
Case 3 and Case 4 together bound the number of deleted edges in Step~\ref{alg2satstep:outball} and 
Step~\ref{alg2satstep:inball}. 
First, we recall the standard CKR analysis~\cite{calinescu2005approximation}, adapted to our situation with directed graphs with two vertices from the same variable {\em coupled} in the algorithm. 
Given a variable $z \in X$, let $\ell(z) \in \{ z, \neg z \}$ be the literal chosen as $\ell$ in Step~\ref{alg2satstep:upcenter}.  (So $\ell(z)$ and $\neg \ell(z)$ become the center of the {\em outball} in Step~\ref{alg2satstep:outball} and {\em inball} in Step~\ref{alg2satstep:inball} respectively).
Order $X = \{ z_1, \dots, z_{n'} \}$ in the increasing order of $\min(d(\ell(z_i), \ell_x), d(\ell_y, \neg \ell(z_i)))$ (ties broken arbitrarily). We are interested in the probability of the event $E_i$ that $(\ell_x, \ell_y)$ is deleted by $B_{out}(v_{\ell(z_i)}, \gamma)$ or $B_{in}(v_{\neg \ell(z_i)}, \gamma)$ in Step~\ref{alg2satstep:outball}. 

First, as $\gamma \in [0.1, 0.11]$ is sampled from an interval of length $0.01$, for fixed $i$, \[
\Pr[(\ell_x, \ell_y) \in \partial^+(B_{out}(v_{\ell(z_i)}, \gamma)) \mbox { or } (\ell_x, \ell_y) \in \partial^-(B_{in}(v_{\neg \ell(z_i)}, \gamma))] \leq 200 d(\ell_x, \ell_y).
\]
The crucial observation is that if the random order sampled in Step~\ref{alg2satstep:random} puts $z_j$ for some $j < i$ before $z_i$, then $E_i$ cannot happen at all; for any value of $\gamma$ that puts $(\ell_x, \ell_y) \in \partial^+(B_{out}(v_{\ell(z_i)}, \gamma))$ or 
$(\ell_x, \ell_y) \in \partial^-(B_{in}(v_{\neg \ell(z_i)}, \gamma))$, which implies $\ell_x \in B_{out}(v_{\ell(z_i)}, \gamma)$ or $\ell_y \in B_{in}(v_{\neg \ell(z_i)}, \gamma)$, the fact that $j < i$ implies $\ell_x \in B_{out}(v_{\ell(z_j)}, \gamma)$ or 
$\ell_y \in B_{in}(v_{\neg \ell(z_j)}, \gamma)$ as well, so either $B_{out}(v_{\ell(z_j)}, \gamma)$ or 
$B_{in}(v_{\neg \ell(z_j)}, \gamma)$
will remove $(\ell_x, \ell_y)$ from the instance. Therefore, $z_i$ must be placed before all $z_1, \dots, z_{i-1}$ in the random ordering in order for $E_i$ to have any chance, which implies $\Pr[E_i] \leq O(d(\ell_x, \ell_y)/i)$ and $\sum_i \Pr[E_i] \leq O(\log (n') \cdot d(\ell_x, \ell_y ))$. 

To exploit the fact that $G$ was constructed from a complete 2SAT instance, let
\[
Z' := \{ z \in X : d(\ell(z), \ell_x) < 0.1 - \eps_0 \mbox{ or }
d(\ell_y, \neg \ell(z)) < 0.1 - \eps_0 \}
\]
be the set of variables $z$ that, when we consider $z$ in the for loop of the algorithm, the edge $(\ell_x, \ell_y)$ is surely removed without being deleted; for instance, if $d(\ell(z), \ell_x) < 0.1 - \eps_0$, since $d(\ell_x, \ell_y) \leq \eps_0$ and $\gamma \geq 0.1$, both $\ell_x, \ell_y$ are in $B_{out}(\ell(z), \gamma)$. 

Note that $Z' = \{ z_1, \dots, z_{|Z'|} \}$ is a prefix set in the ordering defined by the distances. The above argument shows that when $z_i \in Z'$ (i.e., $i \leq |Z'|$), then $\Pr[E_i]$ is now zero instead of $O(d(\ell_x, \ell_y)/i)$. 
Therefore, if $|Z'| \geq \eps_1 n'$ for some $\eps_1 > 0$ to be determined, the analysis can be improved to be
\[
\sum_i \Pr[E_i] \leq \frac{d(\ell_x, \ell_y)}{|Z'| + 1} + \frac{d(\ell_x, \ell_y)}{|Z'| + 2} + \dots + \frac{d(\ell_x, \ell_y)}{n'} \leq O(\log(1 / \eps_1) \cdot d(\ell_x, \ell_y)). 
\]

\paragraph{Case 4: Big variables.}
Call $x$ a {\em big} variable if there exists an edge $(x, y)$ with $d(x, y) \leq \eps_0$ and $|Z'| < \eps_1 n'$. 
We show that the number of big vertices is $O(\sdp / n')$, which implies that the total number of 
edges $(x, y)$ with $d(x, y) \leq \eps_0$ and $|Z'| < \eps_1 n'$ (i.e., the ones not handled in the previous cases) is $O(\sdp)$. 

Fix big $x$ and edge $(x, y)$ with $d(x, y) \leq \eps_0$ and $|Z'| < \eps_1 n'$. 
Consider $z \notin Z'$, which implies $d(\ell(z), \ell_x) > 0.1 - \eps_0$ and $d(\ell_y, \neg \ell(z)) > 0.1 - \eps_0$. The fact that all $x, y, z \in X$ after the preprocessing means that $b_{\ell_x}, b_{\ell_x}, b_{\ell(z)}, b_{\neg \ell(z)}$ all differ additively by at most $0.004$. Therefore, by taking $\eps_0$ small enough, we can conclude that 
\[
d(\ell(z), \ell_x) > 0.1 - \eps_0
\quad \Rightarrow \quad 
\| v_{\ell(z)} - v_{\ell_x} \|_2^2 \geq \frac{(0.1 - \eps_0) - 0.004/4}{8} \geq 0.01
\]
and similarly $\|v_{\ell_y} - v_{\neg \ell(z)} \|_2^2 \geq 0.01$. Of course, $d(\ell_x, \ell_y) \leq \eps_0$ implies that $\| v_{\ell_x} - v_{\ell_y} \|_2^2 \leq 8\eps_0$, so 
\[
\| v_{\ell_x} - v_{\neg \ell(z)} \|_2^2 
\geq 
\| v_{\ell_y} - v_{\neg \ell(z)} \|_2^2 - \| v_{\ell_x} - v_{\ell_y} \|_2^2
\geq 0.01 - 8\eps_0.
\]
Similarly, $\| v_{\ell_y} - v_{\ell(z)} \|_2^2 \geq 0.01 - 8\eps_0$. Therefore, by taking $\eps_0$ small enough, $\| v_{\ell_x} - v_{\ell(z)}\|_2^2$ and $\| v_{\ell_x} - v_{\neg \ell(z)} \|_2^2$ are both at least $0.009$.
Together with the fact that $b_{\ell_x}, b_{\ell(z)}, b_{\neg \ell(z)} \leq [-0.002, 0.002]$, we conclude that all four quantities $d(\ell_x, \ell(z)), d(\ell_x, \neg \ell(z)), d(\ell(z), \ell_x), d(\neg \ell(z), \ell_x)$ are $\Omega(1)$, and since the 2SAT clause corresponding to $\{ x, z \}$, regardless of the literals, puts exactly one edge incident on $\ell_x$, this clause contributes $\Omega(1)$ to $\sdp$. 
Since $|Z'| < \eps_1 n'$, the total SDP contribution from the edges incident on $x$ or $\neg x$ is $\Omega(n')$. So the number of big vertices is $O(\sdp / n')$. 

\paragraph{Conclusion.} 
Therefore, the expected total number of edges deleted is at most the sum of (1) the expected number of edges deleted from the preprocessing step, (2) the number of edges $(\ell_x, \ell_y)$ with $d(\ell_x, \ell_y) \geq \eps_0$, (3) the sum of $O(\log(1/\eps_1) \cdot d(\ell_x, \ell_y))$ over all edges $(\ell_x, \ell_y)$ with $d(\ell_x, \ell_y) < \eps_0$ and $|Z'| \geq \eps_1 n$, and (4) the number of edges incident on the big variables. We argued that each is at most $O(\sdp)$, completing the analysis. Finally, while the algorithm is stated as a randomized algorithm, the standard method of conditional expectation yields a deterministic algorithm with the same guarantee~\cite{lee2019partitioning}.

%% file: Inputs/3csp.tex
\section{\texorpdfstring{$\threesat$}{3SAT}}
\label{sec:3sat}
This section presents a quasi-polynomial time algorithm that decides whether a given complete \threesat instance is satisfiable or not. While it is only a special case of \Cref{thm:kcspassgn}, it contains most of the important ideas. In fact, our algorithm is slightly stronger and enumerates {\em all} the satisfying assignments for the instance. 

\subsection{High-level Plan and Structure of Solutions}

Given a complete instance $(V, \calC)$ for \threesat where each (unordered) triple of variables $(u, v, w)$ corresponds to exactly one clause $(\ell_u, \ell_v, \ell_w) \in \calC$ (where $\ell_u, \ell_v, \ell_w$ are literals for $u, v, w$ respectively), let $\alpha : V \to \{ 0, 1 \}$ be an arbitrary satisfying assignment. Our algorithm proceeds in rounds, and each round makes $\poly(n)$ guesses about $\alpha$. Ideally, one of the guesses will allow us to correctly identify the $\alpha$ values for $\Omega(n)$ variables so that the recursion will yield a quasi-polynomial time algorithm. 

In the case where all $\poly(n)$ guesses fail, we can produce a complete \twosat instance which is guaranteed to be satisfied by $\alpha$. Thus, it suffices to enumerate all the \twosat solutions for the produced instance and see whether they are good for the original \threesat. (Similarly, the enumerating algorithm for \threesat will be used for other CSPs.) 

Of course, the first question is: how many solutions are there? The following simple argument based on VC-dimension gives a clean answer for any complete \kcsp.

    \begin{lemma}\label{lem:satnumberassgn}
        Given a complete \textup{\kcsp} instance $(V, \cC)$, the number of satisfying assignments is at most $O(n^{k-1})$.
    \end{lemma}
    \begin{proof} 
        Let $\cA = \cA(V, \cC)$ be the set of satisfying assignments. 
        Consider the set system with elements $V$ and collection of sets $\cS$, where each satisfying assignment $\alpha_i \in \cA$ creates 
        $S_i := \{v_j \in V \ |\ \alpha_i(v_j) = 1 \} \in \cS$. (I.e., $S_i$ is the support of $\alpha_i$.)

        Recall that a subset $U \subseteq V$ is {\em shattered} by $\cS$ if, for every $U' \subseteq U$, there exists $S \in \cS$ such that $U' = S \cap U$. The {\em VC-dimension} of the set system $(V, \cS)$ is the maximum cardinality $|U|$ of any shattered subset $U \subseteq V$. 
        
        We argue that the VC-dimension of this set system is at most $k-1$. Towards contradiction, suppose that the VC-dimension is at least $k$. Then, there is a set $C\subseteq V$ of size $k$, that is shattered by $\cS$. Since we have a complete instance, the set $C$ must be a clause in $\cC$. This would imply that every possible  partial assignment $\alpha(C)\in \{0,1\}^k$ for $C$ is consistent with at least one satisfying assignment $\alpha_i \in \cA$, which implies that $P_C(\alpha(C)) = \sat$. This is a contradiction to the definition of complete \kcsp. We conclude that the VC-dimension is at most $k-1$.
        
        The Sauer–Shelah lemma \cite{MR0307902, shelah1972combinatorial} shows that the number of sets in a set system with VC-dimension at most $k-1$ is $O(n^{k-1})$. Thus, the number of satisfiable assignments $|\cA|$ is $O(n^{k-1})$.
    \end{proof}

        

    For \twosat, there is an algorithm that enumerates all the solutions in time $O(nT + m)$ where $T$ denotes the number of solutions, even for general instances~\cite{feder1994network}. Combined with Lemma~\ref{lem:satnumberassgn}, we have the following enumeration algorithm for complete \twosat. \aayr{do we need to say something for \ncsp{2}?}

    \begin{corollary}\label{cor:2satassgn}
        Given a complete \textup{\ncsp{2}} instance $(V, \cC)$, the set of all satisfying assignments $\cA(V, \cC)$ can be computed in time $O(n^2)$. 
    \end{corollary}

    \subsection{Main Algorithm}
    Given the enumeration algorithm for \twosat, we now present our main algorithm for \threesat. 


    \paragraph{Overview.} 
    Let $\alpha : V \to \{ 0, 1 \}$ be one of the satisfying assignments we are interested in. 
    For a pair of variables $v_i,v_j\in V$, if $\alpha(v_i) = \alpha(v_j) = 1$, every clause --- out of the total $n-2$ clauses containing the pair --- of the form $(\neg v_i, \neg v_j, \ell)$ must fix the third literal to true. (A similar fact holds for other values of $\alpha(v_i), \alpha(v_j)$.)
    We define this formally as follows:
    \begin{definition}\label{def:3satmaj}
        Given every pair of variables $(v_i,v_j)\in \binom{V}{2}$ and every pair of possible values of $v_i,v_j$, say $(a,b)\in\{0,1\}^2$. We define $n_{i,j,a,b}$ as the total number of clauses of the form $(s_iv_i,s_jv_j, \ell)\in \cC$, for literals $\ell$ for all the variables $v\in V\setminus \{v_i, v_j\}$, where the sign $s_i\in \{ \perp, \neg \}$ (resp. $s_j$) is positive ($\perp$) if $a=0$ (resp. $b=0$), and negative ($\neg$) otherwise. 
        (So once we assign $a$ to $v_i$ and $b$ to $v_j$, the value of $\ell$ is fixed.)
        Moreover, define the corresponding set of variables for each literal $\ell$ for each $(a,b)$ as $N_{i,j,a,b}$.
    \end{definition}
    
    Call $(v_i, v_j)$ a \emph{good pair  with respect to $\alpha$} if the value $n_{i,j,\alpha(v_i),\alpha(v_j)}\geq n/16$. So guessing the values $\alpha(v_i)$ and $\alpha(v_j)$ correctly fixes at least $n/16$ variables with respect to $\alpha$. 
    If we can iterate this process for $O(\log n)$ iterations, we will fix all variables and find $\alpha$. 
    This idea of making progress by guessing values of $O(1)$ variables is similar to the random sampling-based and (the roundings of) the hierarchy-based algorithms developed for \maxcsps and fragile \mincsps, though we also guess $v_i$ and $v_j$ instead of randomly sampling them.
    
    However, unlike the previous results, we may face a situation where guessing values for any pair of variables does not yield enough information to make progress. Our simple but crucial insight is that: {\em the fact that the values of $O(1)$ variables are not enough is indeed valuable information.}
    Formally, consider the case where for every pair $(v_i, v_j)$, we have $n_{i,j,\alpha(v_i), \alpha(v_j)} < n/16$. Thus, $(a', b') := \arg \max_{(a,b)} n_{i,j,a,b}$ cannot be the same as $(\alpha(i), \alpha(j))$. We add a 2SAT constraint $(v_i, v_j)\neq (a',b')$, which is surely satisfied by $\alpha$. Doing the same for all pairs, we get a complete 2SAT instance. From \Cref{cor:2satassgn}, we know that the number of satisfying assignments for a complete 2SAT instance on $n$ variables can be at most $O(n)$, and these can be computed in time $O(n^2)$. Thus, we can enumerate all the satisfying assignments to the created 2SAT instance, and $\alpha$ will be one of them.

    
    \paragraph{Algorithm.} The \Cref{alg:3satassgn} takes in input as the variables $V$, the complete instance clauses $\cC$, and an partial assignment to all the literals $\alpha:V\mapsto\{0,1,\frac12\}$, where for any unfixed variable $v$, $\alpha(v) = \frac12$. Initially, the assignment $\alpha$ is completely unfixed, i.e., $\alpha(v)=\frac12$ for all $v\in V$. It outputs $\cA$, which will eventually be all satisfying assignments to $(V, \cC)$.

    Let $V_U$ be the set of unfixed variables. 
    In Step \ref{alg3satassgnstep:lessvars}, if there are only $O(\log n)$ unfixed variables, we enumerate all assignments and check whether they give a satisfying assignment. This takes time $2^{|V_U|} \cdot \poly(n) = \poly(n)$. 

    The algorithm branches by either guessing one of the good pairs (Step \ref{alg3satassgnstep:guess}) and their assignment, or assuming that there are no good pairs, forms a complete \twosat instance (Step \ref{alg3satassgnstep:2sat}). 

    In Step \ref{alg3satassgnstep:guess}, we guess a pair $(v_i, v_j)$ by branching over all pairs of unfixed variables and guess their optimal assignment $(a,b)$. If the guessed pair is a good pair with respect to some satisfying assignment $\alpha$, i.e., if $n_{i,j,\alpha(v_i),\alpha(v_j)}\geq |V_U|/16$, then we can correctly fix at least $|V_U|/16$ variables $N_{i,j,\alpha(v_i),\alpha(v_j)}$ and solve the problem recursively. Thus we only need to recurse to depth at most $16 \log n$ before fixing all the variables optimally. This is ensured by Step \ref{alg3satassgnstep:lessvars} which optimally computes the assignment to $V_U$ once $|V_U|\leq 100 \log n$.

    In Step \ref{alg3satassgnstep:2sat}, we branch by assuming that there is no good pair with respect to $\alpha$ among the unfixed variables. Under this assumption, we know that $n_{i,j,\alpha(v_i),\alpha(v_j)}< |V_U|/16$. Since, all the four quantities $n_{i,j,a,b}$ over all $a,b\in\{0,1\}$ sum up to $|V_U|-2$, 
    defining $(a', b') := \mathrm{argmax}_{(a,b)\in\{0,1\}}n_{i,j,a,b}$ (while arbitrarily breaking ties) ensures that $(\alpha(v_i), \alpha(v_j)) \neq (a', b')$. 
    So, we add the \twosat constraints $(v_i, v_j)\neq (a', b')$ for all $v_i, v_j\in V$. This gives us a complete \twosat instance $\cC_2$ over the unfixed variables that is guaranteed to be satisfied by $\alpha$. In Steps \ref{alg3satassgnstep:find2sat} and \ref{alg3satassgnstep:check2sat}, we compute and check whether any of the satisfiable assignments of the \twosat instance $\cC_2$ satisfy the $\threesat$ instance $\cC$ or not. We add every such satisfying assignment to the set $\cA$.

    \begin{algorithm}[!htb]
        \caption{ALG-$\threesat$-decision$_{\cA}(V,\cC,\alpha)$}
        \label{alg:3satassgn}
        \begin{algorithmic}[1]
            \Require $V, \cC, \alpha:V\mapsto \{\frac12, 1, 0\}$
            \State $V_U \gets \{v \in V\ |\ \alpha(v)\neq\frac12 \}$ and $V_F \gets V\setminus V_U$\label{alg3satassgnstep:unfix1}
            \Comment{unfixed and fixed variables}
            \If{$|V_U|\leq 100 \log n$} \label{alg3satassgnstep:lessvars}
                \State For each assignment to $\alpha : V_U \to \{ 0, 1\}$, if the entire $\alpha$ satisfies $\calC$, add $\alpha$ to $\cA$. 
            \EndIf
            \For{all $(v_i, v_j)\in \binom{V}{2}$}\label{alg3satassgnstep:guess}\Comment{Guess a pair}
                \For{all $(a,b)\in \{0,1\}^2$}\Comment{Guess the optimal assignment}
                    \If{$n_{i,j,a,b} \geq |V_U|/16$}
                        \State Let $\alpha'$ be the extension of $\alpha$ where $\alpha'(v_i) \leftarrow a, \alpha'(v_j) \leftarrow b$ and \label{alg3satassgnstep:fix}
                        \State $\quad \quad $ Set $\alpha'(v)$ to the only value possible (\Cref{def:3satmaj}), for every $v \in N_{i,j,a, b}$. 
                        \State ALG-\threesat-decision$_{\cA}(V, \cC, \alpha')$.\Comment{Recurse}
                    \EndIf
                \EndFor
            \EndFor
            \State $\cC_2\gets \emptyset$ \Comment{Set of 2SAT constraints} 
            \For{all $(v_i, v_j)\in \binom{V}{2}$}\label{alg3satassgnstep:2sat}\Comment{Otherwise add \twosat constraints}
                \State $(a',b') \gets \arg\max_{(a,b)\in \{0,1\}^2} n_{i,j,a,b}$ \Comment{Find the maximum violating value}
                \State Add \twosat constraint $\cC_2\gets \cC_2\cup ((v_i, v_j)\neq(a',b'))$
            \EndFor
            \State Compute $\cA(V_U, \cC_2)$. \label{alg3satassgnstep:find2sat}\Comment{Compute all satisfiable assignments of complete \twosat instance}
            \For{$\alpha' \in \cA(V_U,\cC_2)$}\label{alg3satassgnstep:check2sat}\Comment{Enumerate over all satisfiable complete \twosat assignments}
                \State $\alpha(V_U)\gets \alpha'(V_U)$ \Comment{Set the partial $\threesat$ assignment based on \twosat assignment}
                \IIf{$\alpha$ satisfies $(V,\cC)$}
                    \State $\cA\gets \cA \cup \{\alpha\}$.\EndIIf\Comment{Store the satisfying assignment in the set}
            \EndFor
        \end{algorithmic}
    \end{algorithm}

    \paragraph{Correctness.} 
    We prove that the set of satisfying assignments $\cA$ of the algorithm contains all the satisfying assignments to the instance.

    \begin{lemma}\label{lem:3satallassgn}
        Every satisfying assignment $\alpha^*$ belongs to $\cA$.
    \end{lemma}
    \begin{proof}
        Consider any satisfying assignment $\alpha^*$. We prove that in each recursive call, 
        (Case 1) if there exists a good pair $(v_i, v_j) \in \binom{V_U}{2}$ with respect to $\alpha^*$, we will correctly fix at least $ |V_U|/16$ values of $\alpha^*$ in one of the branches, and 
        (Case 2) if there is no good pair with respect to $\alpha^*$, we create a complete 2SAT instance satisfied by $\alpha^*$. 
        If this is true, then in at most $ 16\log n$ recursive depth, we would end up guessing $\alpha^*$ in at least one leaf and store $\alpha^*$ in $\cA$.
        
        \paragraph{Case 1: Branching over good-pairs:} Suppose that for some pair $(v_i,v_j)$ is good pair with respect to $\alpha^*$; i.e., $n_{v_i, v_j, \alpha^*(v_i),\alpha^*(v_j)} \geq |V_U|/16$. Then the algorithm would branch over this choice of $v_i,v_j,$ $ \alpha^*(v_i),$ $\alpha^*(v_j)$ as it branches over all the possible choices of good pairs. In that branch, 
        the definition of $N_{v_i,v_j, \alpha^*(v_i),\alpha^*(v_j)}$ in \Cref{def:3satmaj} ensures that 
        Step~\ref{alg3satassgnstep:guess} will correctly fix the $\alpha^*$ values for every variable in $N_{v_i,v_j, \alpha^*(v_i),\alpha^*(v_j)}$. 

        \paragraph{Case 2: Branching over complete \twosat:} In the other case, for every pair $(v_i,v_j)\in \binom{V_U}{2}$, we have that $n_{v_i,v_j,\alpha^*(v_i),\alpha^*(v_j)} < |V_U|/16$, then consider the branch where the algorithm forms a complete \twosat on the unfixed variables $V_U$ in Step \ref{alg3satassgnstep:2sat}. 
        (Note that the pair $v_i, v_j$ can still be a good pair with respect to some other optimal assignment $\alpha^{**}$.)
        We claim that the complete \twosat instance is satisfiable by assignment $\alpha^*(V_U)$. 
        If this were true, then we would be done as we branch over all the possible \twosat assignments to $V_U$, and $\alpha^*(V_U)$ must be one of them. 
        
        Fix any pair $(v_i, v_j) \in \binom{V_U}{2}$. 
        We know that $n_{v_i,v_j, \alpha^*(v_i),\alpha^*(v_j)}< |V_U|/16$. 
        Since, all the four quantities $n_{v_i,v_j,\alpha(v_i),\alpha(v_j)}$, over all $\alpha(v_i),\alpha(v_j)\in\{0,1\}^{2}$, sum up to $|V_U|-2$, the $\max_{\alpha(v_i),\alpha(v_j)\in\{0,1\}^{2}}($ $n_{v_i,v_j,\alpha(v_i),\alpha(v_j)})$ must be at least $\frac{|V_U|-2}{4}\geq\frac{|V_U|}{8}$ (as $|V_U|>\Omega(\log n)$, and $n$ is large enough) by averaging argument. This maximum value does not correspond to the $n_{v_i,v_j, \alpha^*(v_i),\alpha^*(v_j)}$ value. Since this pair $\alpha(v_i, v_j)$ that achieves the maximum is exactly what our added \twosat constraint rules out, $\alpha^*(v_i, v_j)$ will satisfy this constraint. 
    \end{proof}

    \paragraph{Analysis.} The algorithm either branches over all the $4\cdot \binom{|V_U|}{2}$ number of guesses of pairs and their assignments, or formulates the complete 2SAT instance $\cC_2$ on $V_U$, and checks whether any 2SAT assignment satisfies the \threesat instance $\cC$ or not. 
    Since, before each recursive call, we fix at least $ |V_U|/16$ variables, the depth $d$ of the branching tree would be at most $16\log n$ before the remaining number of variables $|V_U|\leq (1-1/16)^d n$ is less than $100 \log n$ (in which case Step \ref{alg3satassgnstep:lessvars} guesses the entire assignment to $V_U$ optimally). 
    Since there are at most $(4n^2)^{16 \log n}$ recursive calls and each call takes $O(n^2)$ time, the total runtime of the algorithm is $O(n^2) \cdot (4n^2)^{O(\log n)} = n^{O(\log n)}$, which yields the following theorem. 
    \begin{theorem}\label{thm:3satassgn}
        There is a $n^{O(\log n)}$ time algorithm to decide whether a complete instance of \textup{\threesat} on $n$ variables is satisfiable or not. Moreover, this algorithm also returns all the $O(n^2)$ satisfying assignments.
    \end{theorem}

%% file: Inputs/kcsp_vecv.tex
\section{\texorpdfstring{\kcsp}{error}}
\label{sec:kcsp}
In this section, we present a quasi-polynomial algorithm for complete $k$-CSP, proving Theorem~\ref{thm:kcspassgn}; again, in addition to deciding the satisfiability of the instance, our algorithm enumerates all the satisfying assignments. 
    Given an instance $(V, \calC)$ where each (unordered) $k$-tuple of variables $\vecv \in \binom{V}{k}$
    corresponds to exactly one clause $C\in \cC$, let $\alpha : V \to \{ 0, 1 \}$ be an arbitrary satisfying assignment. Our algorithm proceeds in rounds, and each round makes $\poly(n)$ guesses about $\alpha$. Ideally, one of the guesses will allow us to correctly identify the $\alpha$ values for $\Omega(n)$ variables so that the recursion will yield a quasi-polynomial time algorithm.

    In the case where all $\poly(n)$ guesses fail, we can produce a complete \ncsp{(k-1)} instance which is guaranteed to satisfy by $\alpha$. Thus, it suffices to enumerate all the \ncsp{(k-1)} solutions for the produced instance and see whether they are good for the original \kcsp.

    We inductively assume that we have the enumeration algorithm for \ncsp{(k-1)}, and we now present our main subroutine for \kcsp. Note that the base case for \ncsp{2} follows from \Cref{cor:2satassgn}.
    
    \paragraph{Overview.}
    Let $\alpha : V \to \{ 0, 1 \}$ be one of the satisfying assignments we are interested in. 
    Consider a $(k-1)$-tuple of variables $\vecv \in \binom{V}{k-1}$ and suppose that the algorithm already fixed the values for $\vecv$ to $\alpha(\vecv)$. 
    Then, every clause $C =(\vecv, u)$--- out of the total $n-k+1$ clauses containing the $(k-1)$-tuple --- such that $P_C(\alpha(\vecv), \alpha'(u))=\sat$ if and only if $\alpha'(u)=\alpha(u)$, must fix the value of $u$ to $\alpha(u)$. We define this formally as follows:
    \begin{definition}[Fixed Variables]\label{def:kcspmaj}
        Given a \textup{\ncsp{k}} instance, any $(k-1)$-tuple of variables $\vecv\in \binom{V}{k-1}$ and any possible assignment of $\alpha(\vecv)\in\{0,1\}^{k-1}$. We define the set of fixed variables as $N_{\vecv,\alpha(\vecv)}=\{v \in V \setminus \vecv \ |\ P_C(\alpha(\vecv), \alpha'(u))=\sat \mbox{ exactly for one value of } \alpha'(u)\in\{0,1\},\ \mbox{ where } C=(\vecv,u)\}$ 
        and the number of fixed variables as $n_{\vecv,\alpha(\vecv)}=|N_{\vecv,\alpha(\vecv)}|$. Moreover, define the corresponding unique assignment as $\alpha_{\vecv,\alpha(\vecv)}(N_{\vecv,\alpha(\vecv)})$.
    \end{definition}

    Call $\vecv\in \binom{V}{k-1}$ a \emph{good $(k-1)$-tuple  with respect to $\alpha$} if the value $n_{\vecv, \alpha(\vecv)}\geq \eps_k n$. So guessing the values $\alpha(\vecv)$ correctly (w.r.t. $\alpha$), fixes at least $\eps_k n$ variables. If we can iterate this process for $\frac{\log n}{\eps_k}$ iterations, we will fix all variables and find $\alpha$.

    Now, consider the case where for every $(k-1)$-tuple $\vecv$, we have $n_{\vecv, \alpha(\vecv)} < \eps_k n$. Thus, $\veca' := \arg \max_{\veca\in\{0,1\}^{k-1}} n_{\vecv, \veca}$ can not be the same as $\alpha(\vecv)$ (for any constant $\eps_k \leq \frac{1}{2^{k+1}})$. We add a \ncsp{(k-1)} constraint $P_\vecv(\alpha''(\vecv))=\unsat$ if and only if $\alpha''(\vecv)=\veca'$, which is surely satisfied by $\alpha$. Doing the same for all $(k-1)$-tuples, we get a complete \ncsp{(k-1)} instance. From \Cref{lem:satnumberassgn}, we know that the number of satisfiable assignments for a complete \ncsp{(k-1)} instance on $n$ variables can be at most $O(n^{k-1})$, and using our induction hypothesis, these can be computed in time $n^{O(\log n)}$. Thus, we can enumerate all the satisfying assignments to the created \ncsp{(k-1)} instance, and $\alpha$ will be one of them.

     \paragraph{Algorithm.} The \Cref{alg:kcspassgn} takes in input as the variables $V$, the complete instance clauses $\cC$, and an partial assignment to all the literals $\alpha:V\mapsto\{0,1,\frac12\}$ where for any undecided variable $v$, $\alpha(v) = \frac12$. Initially, the assignment $\alpha$ is completely undecided, i.e., $\alpha(v)=\frac12$ for all $v\in V$. It outputs $\cA$, which will eventually be all satisfying assignments to $(V, \cC)$.

    The algorithm first computes the set of unfixed variables $V_U = \{v\in V\ |\ \alpha(v) = \frac12\}$. In Step \ref{algkcspassgnstep:lessvars}, if there are only $O(\log n)$ unfixed variables, we enumerate all assignments and check whether they give a satisfying assignment. This takes time $2^{|V_U|} \cdot \poly(n) = \poly(n)$.

    The algorithm branches by either guessing one of the good $(k-1)$-tuples (Step \ref{algkcspassgnstep:guess}) and their assignment, or assuming that there are no good $(k-1)$-tuples, forms a complete \ncsp{(k-1)} instance (Step \ref{algkcspassgnstep:2sat}).

    In Step \ref{algkcspassgnstep:guess}, we guess a $(k-1)$-tuple $\vecv\in\binom{V_U}{k-1}$ by branching over all $(k-1)$-tuples of unfixed variables and guess their optimal assignment $\veca$. If the guessed $(k-1)$-tuple is a good $(k-1)$-tuple with respect to some satisfying assignment $\alpha$, i.e., if $n_{\vecv, \alpha(\vecv)}\geq \eps_k|V_U|$, then we can correctly fix at least $\eps_k |V_U|$ variables $N_{\vecv, \alpha(\vecv)}$ and solve the problem recursively. Thus we only need to recurse to depth at most $\frac{\log n}{\eps_k}$ before fixing all the variables optimally. This is ensured by Step \ref{algkcspassgnstep:lessvars} which optimally computes the assignment to $V_U$ once $|V_U|\leq 100 \log n$.

    In Step \ref{algkcspassgnstep:2sat}, we branch by assuming that there is no good $(k-1)$-tuple with respect to $\alpha$ among the unfixed variables. Under this assumption, we know that $n_{\vecv, \alpha(\vecv)}< \eps_k| V_U|$. Since, all the $2^{k-1}$ quantities $n_{\vecv, \veca}$ over all $\veca\in\{0,1\}^{k-1}$ sum up to at least $|V_U|-k+1$ as there is at least one unsatisfying assignment for each clause. Defining $\veca' := \mathrm{argmax}_{\veca\in\{0,1\}^{k-1}}n_{\vecv, \veca}$ (while arbitrarily breaking ties) ensures that $\alpha(\vecv) \neq \veca'$ for any $\eps_k \leq \frac{1}{2^{k+1}}$. 
    So, we add the \ncsp{(k-1)} constraints $P_\vecv(\alpha''(\vecv))=\unsat$ if and only if $\alpha''(\vecv)=\veca'$, for all $\vecv\in \binom{V_U}{k-1}$. This gives us a complete \ncsp{(k-1)} instance $\cC_{k-1}$ over the unfixed variables that is guaranteed to be satisfied by $\alpha$. In Steps \ref{algkcspassgnstep:find2sat} and \ref{algkcspassgnstep:check2sat}, we compute and check whether any of the satisfiable assignments of the \ncsp{(k-1)} instance $\cC_{k-1}$ satisfy the \ncsp{k} instance $\cC$ or not. We add every such satisfying assignment to the set $\cA$.

        \begin{algorithm}[!htb]
            \caption{ALG-\kcsp-decision$_{\cA}(V,\cC,\alpha)$}
            \label{alg:kcspassgn}
            \begin{algorithmic}[1]
                \Require $V, \cC, \alpha:V\mapsto \{\frac12, 1, 0\}.$
                    \Return 
                \EndIIf
                \State $V_U \gets \{v\ |\ \alpha(v)\neq\frac12, \forall v\in V\}$ and $V_F \gets V\setminus V_U$\label{algkcspassgnstep:unfix1}
                \Comment{unfixed and fixed variables}
                \If{$|V_U|\leq 100 \log n$} \label{algkcspassgnstep:lessvars}
                    \State For each assignment to $\alpha : V_U \to \{ 0, 1\}$, if the entire $\alpha$ satisfies $\calC$, add $\alpha$ to $\cA$. 
                \EndIf
                \For{all $\vecv\in V_U^{k-1} $}\label{algkcspassgnstep:guess}\Comment{Guess a $(k-1)$-tuple}
                    \For{all $\veca\in \{0,1\}^{k-1}$}\Comment{Guess the optimal assignment}
                        \If{$n_{\vecv,\veca} \geq \eps_k|V_U|$}
                            \State Let $\alpha'$ be the extension of $\alpha$ where $\alpha'(\vecv)\gets \veca$ and 
                            \State Set $\alpha'(N_{\vecv, \veca})\gets \alpha_{\vecv, \veca}(N_{\vecv, \veca})$. \Comment{Assign all fixed variables}
                            \State ALG-$\ncsp{k}$-decision$_{\cA}(V, \cC, \alpha')$.\Comment{Recurse}
                        \EndIf
                    \EndFor
                \EndFor
                \State $\cC_{k-1}\gets \emptyset$ \Comment{Set of \ncsp{(k-1)} constraints} 
                \For{all $\vecv\in V_U^{k-1} $}\label{algkcspassgnstep:2sat}\Comment{If no good tuple, add \ncsp{(k-1)} constraints}
                    \State $\veca' \gets \arg\max_{\veca\in \{0,1\}^{k-1} }n_{\vecv, \veca}$ \Comment{Find the maximum fixing value}
                    \State Create \ncsp{(k-1)} constraint $C=\vecv$ such that $P_C(\alpha(\vecv))=\unsat$ iff $\alpha(\vecv)=\veca'$.
                    \State $\cC_{k-1}\gets \cC_{k-1}\cup \{C\}$
                \EndFor
                \State Compute $\cA(V_U, \cC_{k-1})$. \label{algkcspassgnstep:find2sat}\Comment{Compute all satisfiable assignments of complete \ncsp{(k-1)} instance}
                \For{$\alpha(V_U) \in \cA(V_U,\cC_{k-1})$}\label{algkcspassgnstep:check2sat}\Comment{Enumerate over all satisfiable \ncsp{(k-1)} assignments}
            
                    \If{check-assignment($V,\cC,\alpha) =1$}
                        \State $\cA\gets \cA \cup \{\alpha\}$.\Comment{Store the satisfying assignment to set}
                    \EndIf
                \EndFor
            \end{algorithmic}
        \end{algorithm}

        \paragraph{Correctness.} For the correctness of the algorithm, we need to prove that the algorithm outputs the set of all the satisfying assignments to the instance.
        We now prove that the set of satisfying assignments $\cA$ of the algorithm contains all the satisfying assignments to the instance.
    
        \begin{lemma}\label{lem:kcspallassgn}
            Every satisfying assignment $\alpha^*$ belongs to $\cA$.
        \end{lemma}
        \begin{proof}
            Consider any satisfying assignment $\alpha^*$. We prove that in each recursive call, 
            (Case 1) if there exists a good $(k-1)$-tuple $\vecv$ with respect to $\alpha^*$, we will correctly fix at least $ \eps_k|V_U|$ values of $\alpha^*$ in one of the branches, and 
            (Case 2) if there is no good $(k-1)$-tuple, we create a complete \ncsp{(k-1)} instance satisfied by $\alpha^*$. 
            If this is true, then in at most $ \frac{\log n}{\eps_k}$ recursive depth, we would end up guessing $\alpha^*$ in at least one leaf and return $1$ and store $\alpha^*$ in $\cA$.
            
            \paragraph{Case 1: Branching over good-$(k-1)$-tuples:} Suppose that for some $(k-1)$-tuple $\vecv\in\binom{V_U}{k-1}$ is good $(k-1)$-tuple w.r.t $\alpha^*$; i.e., $n_{\vecv, \alpha^*(\vecv)} \geq \eps_k|V_U|$. Then the algorithm would branch over this choice of $\vecv, \alpha^*(\vecv)$ as it branches over all the possible choices of good $(k-1)$-tuples. Then, the algorithm sets $\alpha^*(\vecv)$ as well as $\alpha^*(N_{\vecv, \alpha^*(\vecv)})$ correctly (follows from the \Cref{def:kcspmaj} of fixed variables).

            \paragraph{Case 2: Branching over complete \ncsp{(k-1)}:} If in a recursive call of the algorithm, for no $(k-1)$-tuple $\vecv\in \binom{V_U}{k-1}$, we have that $n_{\vecv,\alpha^*(\vecv)} \geq \eps_k|V_U|$, then consider the branch where the algorithm forms a complete \ncsp{(k-1)} on the unfixed variables $V_U$ in Step \ref{algkcspassgnstep:2sat}. 
            (Note that the $(k-1)$-tuple $\vecv$ can still be a good $(k-1)$-tuple w.r.t. some other optimal assignment $\alpha^{**}$.)
            We claim that the complete \ncsp{(k-1)} instance is satisfiable by assignment $\alpha^*(V_U)$, which is the restriction of $\alpha^*$ to $V_U$.

            If this were true, then we would be done as we branch over all the possible \ncsp{(k-1)} assignments to $V_U$, and $\alpha^*(V_U)$ must be one of them. 
            
            Fix any $(k-1)$-tuple $\vecv\in \binom{V_U}{k-1}$. 
            We know that $n_{\vecv, \alpha^*(\vecv)}< \eps_k|V_U|$. 
            Since, all the $2^{k-1}$ quantities $n_{\vecv,\alpha(\vecv)}$, over all $\alpha(\vecv)\in\{0,1\}^{k-1}$, sum up to at least $|V_U|-k+1$ (as every clause has at least one unsatisfying assignment), the $\max_{\alpha(\vecv)\in\{0,1\}^{3}}n_{\vecv,\alpha(\vecv)}$ must be at least $\frac{|V_U|-k+1}{2^{k-1}}\geq\frac{|V_U|}{2^k}$ (as $|V_U|>\Omega(\log n)$, and $n$ is large enough) by averaging argument. This maximum value does not correspond to the $n_{\vecv, \alpha^*(\vecv)}$ value for any $\eps_k \leq \frac{1}{2{k+1}}$. Since this $(k-1)$-tuple $\alpha(\vecv)$ that achieves the maximum is exactly what our added \ncsp{(k-1)} constraint rules out, $\alpha^*(\vecv)$ will satisfy this constraint. 
        \end{proof}

    \paragraph{Analysis.} The algorithm either branches over all the $2^{k-1}\cdot \binom{|V_U|}{k-1}$ number of guesses of $(k-1)$-tuples and their assignments, or formulates the complete \ncsp{(k-1)} instance $\cC_{k-1}$ on $V_U$, and checks whether any \ncsp{(k-1)} assignment satisfies the \ncsp{k} instance $\cC$ or not. 
    Since, the former fixes at least $\eps_k |V_U|$ variables, the depth $d$ of the branching tree would be at most $\frac{\log n}{\eps_k}$ before the remaining number of variables $|V_U|\leq (1-\eps_k)^d n$ is less than $100 \log n$ (in which case Step \ref{algkcspassgnstep:lessvars} guesses the entire assignment to $V_U$ optimally). 
    For the latter, using our induction hypothesis, we know that all the assignments to $V_U$ satisfying the \ncsp{(k-1)} instance $\cC_{k-1}$ can be found in $n^{O(\log n)}$ time (this bounds the number of assignments as well).
    Since there are at most $(2^{k-1}n^{k-1})^{\frac{\log n}{\eps_k}}$ recursive calls and each call takes $n^{O(\log n)}$ time, the total runtime of the algorithm is $n^{O(\log n)} \cdot (2^{k-1}n^{k-1})^{\frac{\log n}{\eps_k}} = n^{O(\log n)}$, for $\eps_k=\frac{1}{2^{k+1}}$ and any constant $k$ which yields the following theorem.

    \kcspassgn*

%% file: Inputs/induced2sat.tex
\section{\texorpdfstring{\kitwocsp}{error} }
\label{sec:nae}
    \Cref{thm:kcspassgn} for general \kcsp yields a quasi-polynomial time algorithm, and it is open whether it can be made polynomial-time. In this section, we present a subclass of symmetric \kcsp that admits a polynomial-time algorithm. In a symmetric \kcsp parameterized by 
    $S \subsetneq [k] \cup \{ 0 \}$, 
    we represent each clause $C$ as an unordered $k$-tuple of literals $(\ell_1, \dots, \ell_k)$ with 
    and an assignment $\alpha : V \to \{0, 1 \}$ satisfies $C$ if the number of True literals in $(\ell_1, \dots, \ell_k)$ belongs to $S$. (I.e., swapping the values of any two literals does not change the satisfiability.) 
    A few well-studied examples of symmetric \kcsps include \ksat (when $S = [k]$) and \textsc{NAE-$k$-SAT} (when $S = [k - 1]$). 


    The \kitwocsp problem is a special case of symmetric \kcsps, where for each clause $C\in\cC$, the predicate $P_C$ is such that for every $(k-2)$-tuple $\vecv\subset C$ of $k-2$ variables, for every assignment $\alpha(\vecv)\in\{0,1\}^{k-2}$ to $\vecv$, we have that $P_C(\alpha(\vecv), \alpha(C\setminus \vecv)) = \unsat$ for at least one of the assignments $\alpha(C\setminus \vecv)\in\{0,1\}^2$ to the remaining two variables $C\setminus \vecv$. In other words, for any possible assignment to any $k-2$ variables of the clause, the remaining two variables of the clause have to satisfy \ncsp{2} constraints for the clause to be satisfied. 
    
    One simple example of this is \nae. If we fix any one of the variables in a clause to any value $0$ or $1$, the other two variables cannot be both equal (which is a \ncsp{2} constraint) to this value. In fact, every symmetric \ncsp{3} with a predicate other than \nsat{3} is a \itwocsp{3}.


    From \Cref{cor:2satassgn}, we know that deciding whether a complete \ncsp{2} is satisfiable or not, as well as computing all the $O(n)$ satisfying assignments to it can be done in $O(n^2)$ time. Thus, for \kitwocsp, we can pick an arbitrary $(k-2)$-tuple of variables, and guess their optimal assignment by branching for all the $2^{k-2}$ possible assignments. For each of the branches, from the definition, it follows that the instance on the remaining $n-k+2$ variables must be so that there is a \ncsp{2} clause for each pair. This is a complete \ncsp{2} instance. Therefore, in time $O(n^2)$ we can guess whether there is a satisfying assignment for each of the $2^{k-2}$ branches. Thus, we get a total runtime of $2^{k-2}O(n^2)$ which is $O(n^2)$ for any constant $k$. We get the following:
    
    \begin{theorem}\label{thm:induced2csp}
        For any constant $k$, there is a $O(n^2)$ time algorithm to decide whether a given complete \textup{\kitwocsp} instance is satisfiable or not.
    \end{theorem}

    \begin{algorithm}[!htb]
        \caption{ALG-$\kitwocsp$-decision}
        \label{alg:kitwocsp}
        \begin{algorithmic}[1]
            \State Pick an arbitrary $(k-2)$-tuple $\vecv\in \binom{V}{k-2}$. 
            \For{all $\alpha_{\vecv}\in \{0,1\}^{k-2}$}\Comment{Guess the optimal value $\alpha(\vecv)$ of $\vecv$.}
                \State Construct a new \ncsp{2} instance $\cC^{{k-2}}=\binom{V\setminus \vecv}{2}$.
                \State Set all predicates $P_C^{{k-2}}$ for each $C\in \cC^{k-2}$, to all $\sat$ initially.
                \For{all $v_i, v_j\in V\setminus \vecv$} \Comment{For all other $\binom{n-k+2}{2}$ pairs}
                    \ForAll{ $ \alpha_{(v_i, v_j)}\in\{0,1\}^2$}
                    \State Set $\alpha(\vecv)\gets \alpha_{\vecv}, \alpha((v_i, v_j))\gets \alpha_{(v_i, v_j)}$.
                    \If{ $P_{(\vecv, v_i, v_j)}(\alpha(\vecv,v_i, v_j)) = \unsat$}
                    \State $P^{{k-2}}_{(v_i, v_j)}(\alpha((v_i, v_j))=\unsat$.
                    \EndIf
                    \EndFor
                \EndFor
                \If{ALG-\ncsp{2}-decision$(V\setminus \vecv, \cC^{{k-2}})=1$}   \Comment{Solve the complete \ncsp{2} instance}
                    \State \Return $1$
                \EndIf
            \EndFor
            \State \Return 0.

        \end{algorithmic}
    \end{algorithm}

%% file: Inputs/23csp.tex
\section{\texorpdfstring{\krcsp{2}{3}}{error} }
\label{sec:23csp}
    The \krcsp{2}{3} problem is the same as \kcsp with the only difference being the alphabet (labels) $\Sigma =\{0,1,2\}$. The instance is complete, i.e., there is exactly one clause $C=(v_i, v_j)$ for every pair, and the predicate is such that $P_{C}(\alpha(C)) =\unsat$ for at least one assignment $\alpha(C)\in \Sigma^2$.
    Additionally, define $\Sigma_i$ as the set of the alphabet of variable $v_i$ for all $v_i\in V$. If the possible set of labels for each variable $v_i$, $\Sigma_i\subseteq \Sigma$ is such that $|\Sigma_i|<3$, i.e., there are only at most $2$ (instead of total $r$) possible labels for each variable (but not necessarily $0,1$). Then we can re-label and formulate the problem as a \ncsp{2} problem.

    \paragraph{Overview.} 
    We will use this observation for the algorithm, and try to reduce the label set sizes for vertices by at least $1$, thus reducing complete \krcsp{2}{3} to \ncsp{2}. Consider $\alpha$ to be some satisfying assignment. The idea is to define something similar to the fixed variables from \kcsp. But, unlike \kcsp, we may not be able to fix any variables by just guessing assignment to one variable. 
    However, by guessing a variable $v_i$ and the value $\alpha(v_i)$, we can eliminate of all labels $\alpha_j$ for all variables $v_j$ for which there are unsatisfying assignments to the clause $(v_i, v_j)$, i.e., the predicate $P_{(v_i, v_j)}(\alpha(v_i), \alpha_j)=\unsat$. Formally, for any variable $v_i$ and $\alpha\in\Sigma$, define $n_{v_i,\alpha}$ as follows:
    \begin{definition}[Reduced Variables]\label{def:violate31pac}
         Given any variable $v_i$ and $\alpha(v_i)\in\Sigma_i$. Define $N_{v_i,\alpha(v_i)}$ as the tuple of variables and their assignments $(v_j, \alpha(v_j))$, for all $v_j\in V\setminus v_i$ and $\alpha(v_j)\in \Sigma_j$, such that $P_{(v_i, v_j)}(\alpha(v_i), \alpha(v_j))= \unsat$. Let $n_{v_i, \alpha}$ denote the number of all such reduced $v_j$'s.
    \end{definition}
    In other words, if we guess the correct assignment $\alpha(v_i)$ of $v_i$ (w.r.t the satisfying assignment $\alpha$), we would rule out the possibility of labeling $v_j$ with the label $\alpha_j$ for every $(v_j, \alpha_j)\in N_{v_i, \alpha(v_i)}$. Let us call $v_i$ a \emph{good variable} if $n_{v_i,\alpha(v_i)}$ is at least $\eps n$ for some $0<\eps<1/100$. 
    The idea is to guess the good variable and \emph{reduce} the label size for each of the $\eps n$ variables to at most $2$. Since, every time a good variable reduces at least $\eps n$ vertices, we would recurse to a depth at most $O(\log n/\eps)$ before reducing all the variables while fixing all good variables according to the satisfying assignment $\alpha$. If there is no such good variable (w.r.t $\alpha$), then we know that $n_{v_i,\alpha(v_i)}<\eps n$. Because of the complete instance, we know that the sum of $n_{v_i,\alpha}$ over all assignments $\alpha$ is at least $n-1$ as there is at least one unsatisfying assignment per clause. Thus, by just observing the value $\alpha'$ that maximizes $n_{v_i, \alpha}$, we can rule out one label $\alpha'$ for all the vertices that are not good. Finally, we either would have an instance that has at most two possible labels for each variable, this reduces the problem to \ncsp{2} for which we can check the satisfiability in time $\poly(n)$ (\Cref{cor:2satassgn}). 


    \paragraph{Algorithm.} The entire instance $V, \cC, \{\Sigma_i\}_{v_i\in V}$ is global. The algorithm takes in input as the complete \krcsp{2}{3} instance on un-reduced variables $V_U, \cC_U$. Initially, $V_U =V$, $\cC_U=\cC$.

    At any point of time during the algorithm, we always assume that the clauses are such that the predicates are always restricted to the current $\Sigma$. I.e., for each clause $C$, the predicate only maps $P_C:\Sigma_i\times \Sigma_j \mapsto \{\sat, \unsat\}$, while simply ignoring (deleting) all other mappings defined originally $P_C:\Sigma^2\mapsto \{\sat, \unsat\}$.

    In Step \ref{alg31pac:lessvars}, the algorithm first checks if the number of un-reduced variables $V_U$ is less than $O(\log n)$. If it is, then it guesses the assignment to the variables $V_U$ in polynomial time. Note that reducing the label set to just one value is the same as assigning them the one value. This ensures that all the variables $V$ have label sets of size at most $2$ and we get a \ncsp{2} instance. It checks the satisfiability of the instance and concludes accordingly.

    In Step \ref{alg31pac:reduce}, the algorithm then branches assuming that there is one good variable and then guesses the variable and its satisfying assignment by branching overall the possible possibilities of vertices $v_i$ and assignments $\alpha(v_i)$ for which $n_{v_i, \alpha(v_i)}$ is at least $\eps |V_U|$. Then, we perform the following Reduce procedure: 
    \paragraph{Reduce$(V_U, \cC_U, v_i, \alpha(v_i))$:}
    Given $v_i$ and its assignment $\alpha(v_i)$, consider all the tuples of reduced variables and their corresponding labels $N_{v_i, \alpha(v_i)}$. For every $(v_j, \alpha_j)\in N_{v_i, \alpha(v_i)}$, remove the $\alpha_j$ from the label set $\Sigma_j$, remove the vertex $v_j$ from $V_U$. Note that since all the unsatisfying labels are removed, all the clauses $(v_i, v_j)$ are trivially satisfied for any assignment to $v_j$'s. The procedure returns the remaining instance, i.e., the set of un-reduced vertices $V_U$, and clauses $C_U$ between the $V_U$. 
    
    
    Otherwise, in Step \ref{alg31pac:nogood}, it branches assuming that there is no good variable in $V_U$. For all the variables in $v_i\in V_U$, it then reduces them by removing the label $\alpha' = \mathrm{argmax}_{\alpha\in \Sigma_i}n_{v_i,\alpha}$ from $\Sigma_i$. Once all the variables are reduced, we get a \ncsp{2} instance on $V$. In Step \ref{alg31pac:check}, it checks whether this \ncsp{2} instance is satisfiable or not. 

    \begin{algorithm}[!htb]
        \caption{ALG-\krcsp{2}{3}-decision$_{V,\cC, \{\Sigma_i\}_{v_i\in V}}(V_U, \cC_U)$}
        \label{alg:pac31}
        \begin{algorithmic}[1]
            \IIf{for any $v_i\in V, \Sigma_i = \emptyset$}\ \Return $0$. \EndIIf \Comment{No possible labels for $v_i$.}
            \If{$|V_U|\leq 100 \log n$}\label{alg31pac:lessvars}
                    \ForAll{ $\alpha(V_U)\in\{\Sigma_i\}_{v_i\in V_U}$}
                    \State \aay{Make copy $\Sigma'_i\gets \Sigma_i$ for all $v_i\in V$.} 
                    \State For all $v_i\in V_U$, set $\Sigma'_i\gets \{\alpha(v_i)\}$.
                    
                    \State \Return $1$ iff the entire \ncsp{2} instance $(V, \cC, \{\Sigma'_i\}_{v_i\in V})$ is satisfiable.
                    
                    \EndFor
                \State \Return $0$.
            \EndIf
            \For{$v_i\in V_U$}\label{alg31pac:goodpair} \Comment{Guess a good variable}
                \For{all $\alpha(v_i)\in \Sigma_i$} \Comment{Guess the satisfying assignment $\alpha(v_i)$.}
                    \If{$n_{v_i,\alpha(v_i)}\geq \eps |V_U|$}
                        \State Make $\Sigma'_i\gets \Sigma_i$ for all $v_i\in V$ a copy of labels to reset later wards. 
                        \State $\Sigma_i\gets \{\alpha(v_i)\}$.
                        \State $V'_U, C'_U\gets$ Reduce$(V_U, \cC_U, v_i, \alpha(v_i))$.\label{alg31pac:reduce}
                        \State Return $1$ if ALG-\krcsp{2}{3}-decision$(V'_U,\cC'_U)$
                        \State Reset $\Sigma_i\gets \Sigma'_i$ for all $v_i\in V$ to original value before recursion. 
                    \EndIf
                \EndFor
            \EndFor
            \For{$v_i\in V_U$}\label{alg31pac:nogood}\Comment{Branch assuming no good variable}
                \State $\alpha' \gets \arg \max_{\alpha\in\Sigma_i}n_{v_i,\alpha}$. \Comment{Find max unsat constraint}
                \State Reduce $\Sigma_i \gets \Sigma_i \setminus\{\alpha'\}$.\Comment{Reduce variable}
            \EndFor
            \State \Return $1$ iff \ncsp{2} instance $(V,\cC, \{\Sigma_i\}_{v_i\in V})$ is satisfiable.\label{alg31pac:check}
        \end{algorithmic}
    \end{algorithm}

        \paragraph{Correctness.} Now, we proceed to prove the correctness of the algorithm. 
        Consider any satisfying assignment $\alpha$. If the number of un-reduced variables $V_U$ is $O(\log n)$, then in Step \ref{alg31pac:lessvars} we guess the assignment to these correctly w.r.t $\alpha$. Every variable has at most two labels, thus the instance is satisfiable if and only if the resulting \ncsp{2} instance is.

        If there are any good vertices in $V_U$, then in Step \ref{alg31pac:goodpair}, we correctly guess some good vertex's assignment, say $v_i$, $\alpha(v_i)$, reduce the corresponding vertices in $N_{v_i, \alpha(v_i)}$, and then recurse. As we prove below in \Cref{lem:23cspcomplete3red}, this creates an instance that is satisfiable if and only if the $\alpha$ is a satisfying assignment.

        \begin{lemma}\label{lem:23cspcomplete3red}
            Given a complete instance $V_U,\cC_U$, if we correctly fix variable $v_i\in V_U$ to $\alpha(v_i)$ and reduce the other variables, say $V_R$, defined by $N_{v_i, \alpha(v_i)}$, then the instance on remaining un-fixed on vertices $V'=V_U\setminus \{v_i\}$ is satisfiable if and only if the original instance on vertices $V_U$ is satisfied with $\alpha$. Moreover, the instance on the un-reduced vertices $V'_U=V\setminus (\{v_i\}\cup V_R)$ is a complete instance.
        \end{lemma}
        \begin{proof}
            From the definition of $N_{v_i, \alpha(v_i)}$, we know that after assigning $\alpha(v_i)$ to $v_i$, the only clauses that can be unsatisfied correspond to variables (and their unsatisfying labels) from $N_{v_i, \alpha(v_i)}$. Since we ensure in the Reduce procedure that we remove all the unsatisfying labels for these clauses containing $v_i$, we can safely assume that all the corresponding clauses associated with $v_i$ are always satisfied (as there are no unsatisfying pairs now). Thus, the instance on variables $V_U$ is satisfiable if and only if the instance on the remaining variables $V'$ is. No constraints between any two variables $v_j, v_k\in V'_U$ are removed, so the instance is complete.
        \end{proof}

        Otherwise, if there are no good vertices, then as argued earlier we can correctly remove the label $\alpha'=\arg\max_{\alpha\in\Sigma_i}n_{v_i, \alpha}$ for all $v_i\in V_U$. Again, the correctness follows from the following claim:

        \begin{lemma}\label{lem:23cspifandonlyif}
            Given there are no good variables w.r.t $\alpha$, then $\alpha$ is a satisfying assignment to \textup{\krcsp{2}{3}} instance $(V_U,\cC, \{\Sigma_i\}_{v_i\in V})$ if and only if the \textup{\ncsp{2}} instance $(V_U,\cC, \{\Sigma'_i\}_{v_i\in V}))$ is satisfiable, where $\Sigma'_i = \Sigma_i\setminus \{\argmax_{\alpha\in\Sigma_i}n_{v_i, \alpha}\}$.
        \end{lemma}
        \begin{proof}
            Clearly, one way is trivially true. I.e., every satisfying assignment of the \ncsp{2} instance is a satisfying assignment of the \krcsp{2}{3} instance (the former just restricts the labels to a subset of labels of the later, with the same clauses). We now prove that the other way is true as well.

            Given there are no good variables w.r.t $\alpha$, then the \krcsp{2}{3} instance $(V_U,\cC, \{\Sigma_i\}_{v_i\in V})$, we have that for all $v_i\in V_U$, the value $n_{v_i, \alpha(v_i)} < \eps|V_U|$, where $\alpha$ is a satisfying. Since, this is a complete instance $\sum_{\alpha\in \Sigma_i} n_{v_i,\alpha} \geq |V_U|-1$ for all $v_i\in V_U$, and $\max_{\alpha\in \Sigma_i} n_{v_i,\alpha} \geq \frac{n}{3}$. Therefore, we have that $\alpha(v_i)\neq  \arg \max_s{\alpha\in \Sigma_i} n_{v_i,\alpha}$ for all $v_i$. Thus, $\alpha$ satisfies the \ncsp{2} instance.
        \end{proof}

    \paragraph{Analysis.} Every step of the algorithm takes time $n^{O(1)}$. We just need to bound the recursion depth. Since, the before each recursive call, Step \ref{alg31pac:goodpair} reduces at least $\eps |V_U|$ un-reduced variables, the depth $d$ of the branching tree would be at most $\frac{\log n}{\eps}$ before the remaining number of variables $|V_U|\leq (1-\eps)^d n$ is less than $100 \log n$ (in which case Step \ref{alg31pac:lessvars} guesses the entire assignment to $V_U$ correctly according to the satisfying assignment $\alpha$). Thus, the runtime of the algorithm is $n^{O(\log n/\eps)}$. Setting $\eps=1/100$, we get:
    \twothreecsp*

%% file: Inputs/pac.tex
\section{\texorpdfstring{\textsc{PAC}}{error}}\label{sec:pac}
    In the problem of $(r,\ell)$-\textsc{Permutation-Avoiding-Coloring} (\pac{r,\ell)}, we are given a graph on $n$ vertices $V$, edges $\cC\subseteq \binom{V}{2}$, colors (labels) for each vertex defined as $\Sigma=\{0,1,\ldots, r-1\}$ and $1\leq \ell\leq r$. Each edge $(v_i, v_j) \in \cC$ comes along with a matching $\pi_{ij}\subseteq \Sigma\times \Sigma$ of size $\ell$ (and clearly at most $r$), i.e., if $(\alpha_i, \alpha_j)\in \pi_{ij}$, then $(\alpha'_i, \alpha_j)\notin \pi_{ij}$ for any $\alpha'_i\neq \alpha_i \in \Sigma, \alpha_j\in \Sigma$. The goal is to color the graph with coloring $\alpha:V\mapsto \Sigma$ so that for no edge $(v_i, v_j)$ are such that $(\alpha(v_i), \alpha(v_j))\in \pi_{ij}$.  If there is a constraint for every pair of vertices, i.e., if $\cC=\binom{V}{2}$, then we call the instance \emph{complete}. Moreover, if for a complete instance, we relax the condition that each matching is of size exactly $\ell$ to at least $\ell$, then we call the instance \emph{over-complete}.

    The problem $\pac{r,r}$, also simply called $r$-\textsc{PAC}, is a generalization of $r$-\textsc{Coloring} where each edge can rule out any permutation between colorings of the two vertices (instead of just the identity permutation). 

    Note that the problem of complete \pac{r,\ell'} reduces to over-complete \pac{r, \ell}, for all $\ell'\geq \ell$. Moreover, it is easy to see that \pac{r, \ell} is a special case of \krcsp{2}{r}, where for every clause $C=(v_i, v_j)\in\cC$, we have the set of all unsatisfying assignments $\pi_{ij}=\{ (\alpha_i, \alpha_j)\in\Sigma^2\ |\ P_C(\alpha_i, \alpha_j)=\unsat\}$ 
    such that they form a matching in $\Sigma\times \Sigma$. Additionally, $|\pi_{ij}| = \ell$. Similarly, a complete (or over-complete) instance of \pac{r, \ell} reduces to a complete instance of \krcsp{2}{r} as there are are exactly $\ell \geq 1$ (or at least $\ell \geq 1$ for over-complete) $\unsat$ constraints for each pair $v_i, v_j\in V$. 
    This is captured by the following observation:

    \begin{observation}\label{obs:pactocsp}
        For any $1\leq \ell\leq r, r\geq 2$, over-complete \textup{\pac{r,\ell}} reduces to complete \textup{\krcsp{2}{r}}. And, for every $\ell'\geq \ell$, complete \textup{\pac{r,\ell'}} reduces to over-complete \textup{\pac{r,\ell}}.
    \end{observation}

    Using \Cref{obs:pactocsp} and \Cref{thm:23csp}, we get that the problem of over-complete (and therefore complete as well) $\pac{3,\ell}$, for any valid $\ell$, has a $n^{O(\log n)}$ decision algorithm.

    Now, we use similar ideas of reducing the label set size to obtain quasi-polynomial time algorithms to decide whether an over-complete instance of \pac{4,3} (and therefore, complete \pac{4,4}), and complete instance of \pac{5,5} is satisfiable or not.

    \subsection{\texorpdfstring{\pac{4,3}}{error}}

        \paragraph{Overview.} Similar to \krcsp{2}{3}, we define individual label sets to each vertex $v_i\in V$ as $\Sigma_i$. Initially, all the $\Sigma_i =\Sigma$. The idea is to fix some good vertices and reduce vertices, similar to that of  \krcsp{2}{3}. For any vertex $v_i$ and $\alpha(v_i)\in\Sigma$, define $n_{v_i,\alpha(v_i)}$ as follows:
    \begin{definition}[Reduced vertices]\label{def:violate43pac}
         For any vertex $v_i$ and $\alpha(v_i)\in\Sigma$, define $n_{v_i,\alpha(v_i)}$ as the number of vertices $v_j\in V\setminus \{v_i\}$ such that $(\alpha(v_i), \alpha(v_j))\in \pi_{ij}$ for any values of $\alpha(v_j)\in \Sigma$. Let $N_{v_i, \alpha(v_i)}$ denote the set of all such tuples $(v_j, \alpha(v_j))$.
    \end{definition}
        In other words, if we set $v_i$ to $\alpha(v_i)$, we would rule out the possibility of labeling $v_j$ with the label(s) $\alpha_j$ for every $(\alpha_i,\alpha_j)\in\pi_{ij}$. Let us call $v_i$ a \emph{good vertex} if $n_{v_i,\alpha^*}$ is at least $\eps n$ for some $0<\eps<1/100$, where $\alpha^*$ is some satisfying assignment. The idea is to guess the good vertex optimally (w.r.t this satisfying assignment $\alpha^*$) and \emph{reduce} the label set size for each of the $\eps n$ vertices to $3$. 
        Since, every time a good vertex reduces at least $\eps n$ vertices, we would recurse at to depth at most $O(\log n/\eps)$ before finding all the satisfying assignments. If there is no such good vertex, then we know that $n_{v_i,\alpha^*}<\eps n$. Because of the over-complete instance, we know that the sum of $n_{v_i,\alpha}$ over all assignments $\alpha$ is at least $n-1$. Thus, by just observing the value $\alpha'$ that maximizes $n_{v_i, \alpha}$, we can rule out one label $\alpha'$ for all the vertices that are not good. This ensures that the label set size is reduced by one for each vertex. Moreover, we will have at least $1$ unsatisfying assignment between each pair (as deleting one label each from a vertex in a pair, deletes at most $2$ unsatisfying assignments). This reduces the problem to deciding the satisfiability of over-complete \pac{3,1}, which takes $n^{O(\log n)}$ time. 

        However, while reducing the vertices it might happen that we have to reduce the labels of some vertices to less than three as well. This could give us an over-complete \pac{3,1} instance on the vertices with exactly $3$ labels along with extra general \ncsp{2} constraints to satisfy the vertices with at most $2$ labels. We can always assume that there are no vertices with $0$ labels, otherwise the instance is clearly unsatisfiable. This is not a problem as the \Cref{alg:pac31} works as is on such instances, and reduces the entire thing to a general \ncsp{2} instance. 

        \paragraph{Algorithm.} The \Cref{alg:pac43} takes in input as the instance $V, \cC, \{\Sigma_i\}_{v_i\in V}$. Initially, $\Sigma_i = \Sigma$ for all $i\in V$. At any point in time during the algorithm, we always assume that every clause $(v_i, v_j)\in \cC$ is such that the $\pi_{ij}$ is always restricted to the current $\Sigma$. I.e., for each clause $(v_i, v_j)\in \cC$, the matching is only on $\pi_{ij}\subseteq \Sigma_i\times \Sigma_j$, while simply ignoring (deleting) all other mappings defined originally $\pi_{ij}\subseteq \Sigma^2$.

        Similar to the \Cref{alg:pac31}, we use a Reduce subroutine that given a variable $v_i$ and its guessed optimal assignment $\alpha(v_i)$, reduces the labels of all the vertices $v_j$ with any unsatisfying assignments $(\alpha(v_i), \alpha_j)$ to clauses $(v_i, v_j)$ by deleting the corresponding label $\alpha_j$ from $\Sigma_j$. Intuitively, we do this to ensure that there are no constraints to be taken care of once a variable a fixed to some specific value so that we can forget about the fixed variable and iterate on the remaining instance.  Formally, we define the procedure as follows:
        
            \paragraph{Reduce$(V, \cC, \{\Sigma_i\}_{v_i\in V}, v_i, \alpha(v_i))$:} Given $v_i$ and its optimally guessed assignment $\alpha(v_i)$, consider all the tuples of reduced variables and their corresponding labels $N_{v_i, \alpha(v_i)}$. For every $(v_j, \alpha_j)\in N_{v_i, \alpha(v_i)}$, remove the $\alpha_j$ from the label set $\Sigma_j$.

        The algorithm first computes the set of un-reduced vertices $V_U$ with exactly $4$ possible labels, i.e., $\Sigma_i=\Sigma$ for all $v_i\in V_U$. In Step \ref{alg31pac:lessvars}, the algorithm checks if the number of un-reduced vertices $V_U$ is less than $O(\log n)$. If it is, then it guesses the assignment to the vertices $V_U$ in polynomial time. Note that reducing the label set to just one value is the same as assigning them the one value. 
        Then, we reduce the instance based on the assignment of the variables in $V_U$ using the Reduce procedure.
        This ensures that all the vertices $V$ have label sets of size at most $3$ and we get an over complete \pac{3,1} instance on $V_3$ with additional \ncsp{2} constraints given by $V_2\cup V_1$. It checks the satisfiability of the instance and concludes accordingly. We formally define this subroutine as follows:

        \paragraph{Check-\krcsp{2}{3}$(V, \cC, \{\Sigma_i\}_{v_i\in V})$:} The procedure takes in the input set of the vertices with all the label sets of size at most $3$. Then, it first checks if there are any vertices with an empty label set. If so, it returns that the instance is unsatisfiable. If not, it computes the partitions $V_3, V_2$ of the vertex set $V$, containing vertices of label set size three and at most two respectively. Then, it checks the satisfiability of the over-complete \pac{3,1} on $V_3$ along with the \ncsp{2} constraints from $V_2$ using the \Cref{alg:pac31}. Finally, it returns $1$ if the instance is satisfiable and $0$ otherwise.

        In Step \ref{alg43pac:goodpair}, the algorithm then branches assuming that there is at least one good vertex among the un-reduced vertices in $V_U$ and then guesses the vertex and its satisfying assignment by branching over all the possibilities of vertices $v_i$ and assignments $\alpha(v_i)$ for which $n_{v_i, \alpha(v_i)}$ is at least $\eps |V_U|$. Then, it reduces the label sets of all the vertices in tuples from $N_{v_i, \alpha(v_i)}$ using the Reduce procedure. Then, it recurses with at least $\eps|V_U|$ more vertices that are reduced. Note, that the the number of un-reduced vertices decreases by $\eps |V_U|$ in each recursive call, so we recurse at most $\log n/\eps$ times. 
    
        Otherwise, in Step \ref{alg43pac:nogood}, it branches assuming that there is no good vertex in $V_U$. For all the vertices in $v_i\in V_U$, it then removes the label $\alpha' = \max_{\alpha\in \Sigma_i}n_{v_i,\alpha}$ from $\Sigma_i$. Once all the vertices have label set size at most $3$, as before, we call the procedure Check-\krcsp{2}{3}, to check the satisfiability of the over-compelete \pac{3,1} instance along with general \ncsp{2} constraints.

        \begin{algorithm}[!htb]
            \caption{ALG-\pac{4,3}-decision$(V,\cC,\{\Sigma_i\}_{v_i\in V})$}
            \label{alg:pac43}
            \begin{algorithmic}[1]
                \State Compute $V_U \gets \{v\in v\ |\ |\Sigma_v|=4 \}$.
                \If{$|V_U|\leq 100 \log n$}\label{alg43pac:lessvars}
                    \ForAll{$v_i\in V_U$}\Comment{Iteratively guess and reduce for all vertices in $V_U$.}
                    \State Guess the optimal value $\alpha(v_i)$ of $v_i$. $\Sigma_i\gets \{\alpha(v_i)\}$.
                    \State Reduce$(V, \cC,\{\Sigma_i\}_{v_i\in V},v_i, \alpha(v_i))$.
                    \EndFor
                    \State \Return $1$ iff Check-\krcsp{2}{3}$(V\setminus V_U, \cC, \{\Sigma_i\}_{v_i\in V})$.
                \EndIf
                \For{$v_i\in V_U$}\label{alg43pac:goodpair} \Comment{Guess a good vertex}
                    \For{all $\alpha(v_i)\in \Sigma_i$} \Comment{Guess the satisfying assignment $\alpha(v_i)$.}
                        \If{$n_{v_i,\alpha(v_i)}\geq \eps |V_U|$}
                            \State Create $\Sigma'$ copy of $\Sigma$.
                            \State $\Sigma'_i\gets \{\alpha(v_i)\}$.
                            \State Reduce$(V, \cC,\{\Sigma'_i\}_{v_i\in V},v_i, \alpha(v_i))$.
                            \State \Return $1$ if ALG-\pac{4,3}-decision$(V\setminus \{v_i\},\cC, \{\Sigma'_i\}_{v_i\in V})$
                        \EndIf
                    \EndFor
                \EndFor
                \For{$v_i\in V_U$}\label{alg43pac:nogood}\Comment{Branch assuming no good vertex}
                    \State $\alpha' \gets \arg \max_{\alpha\in\Sigma_i}n_{v_i,\alpha}$. \Comment{Find max unsat assignment}
                    \State Set $\Sigma_i \gets \Sigma_i \setminus\{\alpha'\}$.\Comment{Reduce vertex}
                \EndFor
                \State \Return $1$ iff Check-\krcsp{2}{3}$(V, \cC, \{\Sigma_i\}_{v_i\in V})$.\label{alg43pac:check}
            \end{algorithmic}
        \end{algorithm}

        \paragraph{Correctness. }

        We first prove two helper lemmas used to establish the correctness of the algorithm. In the first lemma, we prove that if we fix a variable $v_i$ optimally to $\alpha(v_i)\in \Sigma_i$, i.e., reduce the label set $\Sigma_i=\{\alpha(v_i)\}$, then we can forget about the vertex. I.e., the instance on the remaining vertices without $v_i$ is satisfiable if and only if the original instance including $v_i$ is.
        \begin{lemma}\label{lem:43pacthree}
            Given instance $V,\cC$, and any satisfying assignment $\alpha$, if we fix variable $v_i\in V_U$ correctly according to $\alpha(v_i)$ and reduce the other variables corresponding to $N_{v_i, \alpha(v_i)}$, then the instance on the remaining unfixed vertices $V'=V\setminus \{v_i\}$ is satisfiable if and only if the $\alpha$ is a satisfying assignment to the original instance on vertices $V$.
        \end{lemma}
        \begin{proof}
            From the definition of $N_{v_i, \alpha(v_i)}$, we know that after assigning $\alpha(v_i)$ to $v_i$, the only clauses that can be unsatisfied correspond to variables (and their unsatisfying labels) from $N_{v_i, \alpha(v_i)}$. Since we ensure in the Reduce procedure that we remove all the unsatisfying labels for these clauses containing $v_i$, we can safely assume that all the corresponding clauses associated with $v_i$ are always satisfied (as there are no unsatisfying pairs now). Thus, the instance on variables $V$ is satisfiable if and only if the instance on the remaining variables $V'$ is.
        \end{proof}

        Next, in the second lemma, we prove that after any point in the algorithm, if we have label sets of all the vertices of size at most $3$, then the resulting instance is an over-complete instance on the vertices with exactly $3$ possible labels (with some additional \ncsp{2} constraints from the other vertices).
        \begin{lemma}\label{lem:43pacpac31}
            Given an instance $V,\cC$ such that all the vertices have label set size at most $3$, then the instance on $V_3=\{v_i\in V\ |\ |\Sigma_i|=3\}$ is an over-complete \textup{\pac{3,1}} instance.
        \end{lemma}
        \begin{proof}
            The proof is simple. Every vertex that has three labels has to be reduced exactly once. Thus, for every pair $v_i, v_j \in V_3$, we can delete at most two unsatisfying constraints from $\pi_{ij}$, leaving behind at least one constraint in $\pi_{ij}$ restricted to the new label sets.
        \end{proof}

        Now, we proceed to prove the correctness of the algorithm. Consider any satisfying assignment $\alpha$. If the number of un-reduced variables $V_U$ is $O(\log n)$, then in Step \ref{alg43pac:lessvars} we guess the assignment to these correctly w.r.t $\alpha$ and use the Reduce procedure. As we prove in \Cref{lem:43pacthree}, this creates an instance that is satisfiable if and only if the original instance is. Moreover, as we guess the assignment of all the vertices in $V_U$ correctly w.r.t. the satisfying assignment $\alpha$, we know that the remaining instance has label sets of all vertices of size at most three. As we prove in \Cref{lem:43pacpac31}, this instance is an over-complete \pac{3,1} instance with some additional vertices with \ncsp{2} constraints, so we can use the Check-\ncsp{2,3} procedure to check its satisfiability. 

        If there are any good vertices in $V_U$, then in Step \ref{alg43pac:goodpair}, we correctly guess some good vertex's assignment, say $v_i$, $\alpha(v_i)$, reduce the corresponding vertices in $N_{v_i, \alpha(v_i)}$, and then recurse. Again, as we prove in \Cref{lem:43pacthree}, this creates an instance that is satisfiable if and only if the $\alpha$ is a satisfiable assignment.

        Otherwise, if there are no good vertices, then as argued earlier we can correctly remove the label $\alpha'=\arg\max_{\alpha\in\Sigma_i}n_{v_i, \alpha}$ for all $v_i\in V_U$. Again, as we prove in \Cref{lem:43pacpac31}, this instance is an over-complete \pac{3,1} instance with some additional vertices with \ncsp{2} constraints, so we can use the Check-\ncsp{2,3} procedure to check its satisfiability. This concludes the correctness of the algorithm.

        \paragraph{Analysis.} Every recursive call of the algorithm takes time at most $n^{O(1)}$. We just need to bound the recursion depth. Since, the before each recursive call, Step \ref{alg43pac:goodpair} reduces at least $\eps |V_U|$ un-reduced variables, the depth $d$ of the branching tree would be at most $\frac{\log n}{\eps}$ before the remaining number of variables $|V_U|\leq (1-\eps)^d n$ is less than $100 \log n$ (in which case Step \ref{alg43pac:lessvars} guesses the entire assignment to $V_U$ correctly according to the satisfying assignment $\alpha$). From \Cref{thm:23csp}, we know that each call to the procedure Check-\krcsp{2}{3}, that checks the satisfiability of over-complete \pac{3,1} (with additional \ncsp{2} constraints) takes time $n^{O(\log n)}$. Thus, the runtime of the algorithm is $n^{O(\log n/\eps)}$. Setting $\eps=1/100$, we get:
    
        \begin{theorem}
            There is an $n^{O(\log n)}$ time algorithm to decide whether a over-complete instance of \textup{\pac{4,3}} is satisfiable or not.
        \end{theorem}
        \begin{corollary}
            There is an $n^{O(\log n)}$ time algorithm to decide satisfiability of complete-\textup{\pac{4,3}}, complete-\textup{\pac{4,4}}. 
        \end{corollary}

    \subsection{\texorpdfstring{\pac{5,5}}{error}}

    We briefly describe the algorithm for the complete \pac{5,5} problem, where there are exactly $5$ unsatisfying constraints, i.e., the matching $\pi_{ij}$ is such that $|\pi_{ij}|=5$. The intuition is that once we guess the label for any arbitrary vertex correctly according to a satisfying assignment $\alpha$, we can rule out exactly one label for all other vertices, thus reducing to over-complete \pac{4,3}.

    \paragraph{Algorithm.} The Algorithm is fairly simple. Let $\alpha$ be any satisfying assignment. Consider any arbitrary vertex $v_i\in V$, and guess the value $\alpha(v_i)$ of $v_i$ by branching over all possible values of $\Sigma$. Remove the label $\alpha_j$ such that $(\alpha(v_i), \alpha_j)\in \pi_{ij}$ for every other $v_j \in V\setminus \{v_i\}$. Since, $|\pi_{ij}|=5$, we know that there is exactly one such label $\alpha_j$ for each $v_j$. Now, check the satisfiability of the over-complete \pac{4,3} instance on $V\setminus \{v_i\}$ and conclude accordingly.

    The correctness of the algorithm follows trivially. Every time we remove one label from each vertex, for each pair $v_j, v_k\in V\setminus \{v_i\}$ we could delete at most two unsatisfying constraints from the total of five. Thus, every pair has at least $3$ constraints. This reduces to a complete \pac{4,3} instance, and we can check the satisfiability of this instance in time $n^{O(\log n)}$, giving us the following result:
    
    \begin{theorem}
        There is an $n^{O(\log n)}$ time algorithm to decide whether a complete instance of \textup{\pac{5,5}} is satisfiable or not.
    \end{theorem}

%% file: Inputs/hardness.tex
\section{Hardness}
\label{sec:hardness}
In this section, we prove various hardness results claimed in this paper.

\subsection{Hardness of \texorpdfstring{\ksat}{kSAT} on Dense Instances}
We start by showing that (\textsc{Min}-)\ksat on dense instances is almost as hard as general instances. 

\begin{claim}
For any constant $\eps > 0$, there is an approximation-preserving polynomial-time reduction from a general instance of \textup{\minksat} to an instance whose constraint graph has at least $(1 - \eps) \binom{n}{k}$ constraints.
\label{claim:hardness-dense}
\end{claim}
\begin{proof}
Given a general instance of \minksat with variable set $V_0$ with $n_0 = |V_0|$, add $O(n_0 / \eps)$ dummy variables. For every $k$-set of variables $(v_1, \dots, v_k)$ including at least one dummy variable, create a dummy constraint $(v_1 \vee \dots \vee v_k)$. Output the original instance combined with the dummy variables and constraints. 
The number of constraints is at least $\binom{n_0 + n}{k} - \binom{n_0}{k} \geq (1 - \eps_0)\binom{n_0 + n}{k}$. 
For any assignment of the original variables, one can easily satisfy all the dummy constraints by setting all the dummy variables True. In the other direction, for any assignment of the new instance, changing all dummy variables to True will only satisfy more constraints, which will possibly violate only the original constraints. 
Therefore, the optimal values of the two instances are the same.
\end{proof}

\begin{remark}
Note that the above instance is {\em everywhere} dense in the sense that each variable is contained in at least $(1 - \eps)\binom{n}{k-1}$ constraints. Also, though we do not formally define high-dimensional expansion in this paper, we believe that this instance has almost maximum expansion in every definition of high-dimensional expansion. 
\end{remark}

\subsection{Hardness of \mincsp}
We prove that for many (but not all) CSPs, the hardness of exact optimization (which does not distinguish max/min) implies the hardness in complete instances. 
Given $k \geq 2, r \geq 2$ and the alphabet $\Sigma$ with $|\Sigma| = r$, and the constraint family $\Gamma = \{ P_1, \dots, P_t \}$ with each $P_i : \Sigma^k \to \{ \sat, \unsat \}$, recall that $\mincsp(\Gamma)$ is a CSP where each predicate is from $\Gamma$. 

\begin{theorem}
Suppose that \textup{$\mincsp(\Gamma)$} on general instances is NP-hard and there is a distribution $\calD$ supported on $\Gamma$ such that for every $\sigma \in \Sigma^k$, $\Pr_{P \sim \calD}[\sigma \notin P] = p$ for some $p \in (0, 1)$. Then, there is no polynomial-time algorithm for complete \textup{\minkcsp} unless $\mathbf{NP} \subseteq \mathbf{BPP}$. 
\label{thm:hardness-min}
\end{theorem}
While the condition about $\calD$ is technical, it is easily satisfied by many CSPs like \ksat, \klin, \kand, and \ug. One notable CSP that does not meet this condition is $r\textsc{-Coloring}$ (not permutation avoiding), which is consistent with the fact that it is hard on general instances, but not on complete ones. 
\begin{proof}
Given a general instance for $\mincsp(\Gamma)$, described by $V$, $\calC$, and $\{ P_C \}_{C \in \calC}$, let $t = \poly(n)$ be a parameter to be determined. Let $m = |\calC|$. 
Our new instance is constructed as follows. 
\begin{itemize}
\item Variables: $V' := V \times [t]$. 
\item For each original constraint $C = (v_1, \dots, v_k)$ with $P_C$ and $i_1, \dots, i_k \in [t]$, create a constraint $((v_1, i_1), \dots, (v_k, i_k))$ with $P_C$. Call them {\em real} constraints, let $\calC'_r$ be the set of them, and $m'_r := |\calC'_r|$. 
\item For every other $k$-set $C$ of $V'$, independently sample $P_C$ from $\mathcal{D}$. Call them {\em dummy} constraints, let $\calC'_d$ be the set of them, and $m'_d := |\calC'_d|$. 
\end{itemize}
We would like to show that the number of dummy constraints unsatisfied is tightly concentrated over all assignments $\alpha$. 
Fix any assignment $\alpha : V' \times \Sigma$. Then the expected number of dummy constraints unsatisfied is $p m'_d$. The Chernoff bound ensures that for any $\eps > 0$, the probability that it deviates from the expected value by more than $\eps p m'_d$ is $e^{-\Omega(\eps^2 p m'_d)}$. 
Since $m'_d \in [\binom{n}{k-1} t^k , (nt)^k]$ and $p$ is a constant only depending on $\Gamma$, letting $\eps = 1/n^{k+1}$ ensures that $\eps p m'_d \leq t^k / n$ and $e^{-\eps^2 p m'_d} = e^{-\Omega(t^k / n^{k+3})}$. Since there are at most $e^{tn \log |\Sigma|}$ assignments, 
if we take $t = n^{5k}$, the union bound shows for any $\alpha$, the total number of dummy constraints satisfied by $\alpha$ does not deviate from the expected value $\eps p m'_d$ by more than $t^k / n$. 

For real constraints, given an assignment $\alpha$, let $f(\alpha)$ be the number of real constraints unsatisfied by $\alpha$. Note that $f(\alpha) / t^k$ is exactly the expected number of constraints of original constraints $\calC$ unsatisfied by the random assignment where each $v \in V$ chooses a random $i \in [t]$ and uses $\alpha(v_{i})$. 

Therefore, if the optimal value (number of unsatisfied constraints) in the original general instance is at most $q$, then there exists an assignment in the new instance (where each $v_{i}$ follows the assignment of $v \in V$) that unsatisfied at most $qt^k \pm t^k/n$. 
Otherwise, for any assignment $\alpha$, the number of unsatisfied constraints is at least $(q + 1)t^k \pm t^k/n$. Therefore, the hardness of the general version implies the hardness of the complete version. 
\end{proof}

\subsection{Hardness of \texorpdfstring{\krcsp{2}{4} and \krcsp{3}{3}}{error}}

In this subsection, we prove that \krcsp{2}{4}, 6-\textsc{PAC}, and \krcsp{3}{3} are NP-hard on complete instances, proving 
Theorem~\ref{thm:24csp}, Theorem~\ref{thm:hardness-6pac}, and Theorem~\ref{thm:threethreecsp} respectively.
At a high level, these three reductions start from the hardness of 3-\text{Coloring} and \threesat on general instances (which are \csp{(2,3)} and \csp{(3, 2)} respectively), {\em add the dummy label}, but {\em add constraints forcing the assignment not to use the dummy label}. 

First, the following hardness of \pac{4,1} implies the hardness of \krcsp{2}{4}, proving Theorem~\ref{thm:24csp}. 

\begin{theorem}
\textup{\pac{4, 1}} is NP-hard on complete instances.
\end{theorem}
\begin{proof}

We reduce from $3\textsc{-Coloring}$ (with the color set $\{ 1, 2, 3 \}$).
We want the following gadget, which is an instance of \pac{4, 1}. 
\begin{itemize}
\item Variable is $\{ u, v \} \times [t]$, for some $t = O(1)$ to be determined. Let $U := \{ u \} \times [t]$ and $V := \{ v \} \times [t]$. 

\item For any $i, j \in [t]$, pick a random $p \in [4]$ and put a constraint between $(u, i)$ and $(v, j)$ ruling out the assignment that gives $p$ to both $i$ and $j$.

\item For any pair $i < j$, with probability $1/2$, sample $p \neq q \in [3]$ and put a constraint ruling out $(u, i)$ getting $p$ and $(u, j)$ getting $q$. Otherwise, put a constraint ruling out $(u, i)$ and $(u, j)$ both getting $4$. 
Put the same constraint for $(v, i)$ and $(v, j)$ as well. 
\end{itemize}

This gives a distribution over instances $I$.
We show that at least one instance satisfies the following properties:
\begin{enumerate}
    \item $I$ is not satisfied by any assignment where either $U$ or $V$ has color $4$ appearing $0.1t$ times. 
    
    \item $I$ is not satisfied by any assignment where $U$ or $V$ has two different colors appearing at least $0.1t$ times. 
    
    \item If an assignment does not satisfy the above two conditions, it means that both $u$ and $v$ have one {\em dominant} color in $[3]$ that is used at least $0.7t$ times. If $u$ and $v$ have the same dominant color, then this assignment does not satisfy $I$. 

\end{enumerate}
Given an assignment to the variables, if it does not satisfy any of the above properties, the probability that it satisfies $I$ is at most $e^{-\Omega(t^2)}$. 
Since there are at most $e^{O(t)}$ assignments, for some constant $t = O(1)$, a fixed gadget $I$ with the above properties exists. 

Given a hard instance $G = (V, E)$ for $3\textsc{-Coloring}$, create an instance of \pac{4,1} whose vertices are $V \times [t]$ and put the above gadget for every $(u, v) \in E$. (The gadget inside $\{ u \} \times [t]$ does not depend on $v$ and can be created only once.) For any pair of variables that do not have a constraint, put the \pac{4,1} constraint ruling out both variables getting $4$. 

If there is a valid $3$-coloring of $V$, then its natural extension to $V \times [t]$ (i.e., for $v \in V$ and $i \in [t]$, $v_i$ is assigned the color of $v$) is a valid \pac{4, 1} assignment, and given any valid assignment for the \pac{4,1} instance, by choosing the dominant color for every $v \in V$ yields a valid $3$-coloring. 
\end{proof}

Another application of this strategy yields the hardness of 6-\textsc{PAC}.

\begin{theorem}
6-\textsc{PAC} is NP-hard on complete instances.
\label{thm:hardness-6pac}
\end{theorem}
\begin{proof}

We reduce from $3\textsc{-Coloring}$ (with the color set $\{ 1, 2, 3 \}$).
We want the following gadget, which is an instance of 6-\textsc{PAC}.
\begin{itemize}
\item Variable is $\{ u, v \} \times [t]$, for some $t = O(1)$ to be determined. Let $U := \{ u \} \times [t]$ and $V := \{ v \} \times [t]$. 

\item For any $i, j \in [t]$, between $(u, i)$ and $(v, j)$, put the constraint ruling out $\{ ( a, a) \}_{a \in [6]}$.

\item For any pair $i < j$, between $(u, i)$ and $(u, j)$, and put a constraint ruling out the pairs $\{ (1, a), (2, b), (3, c), (4, 4), (5, 5), (6, 6) \}$ where $a, b, c \in [3]$ are sampled as three distinct numbers conditioned on $a \neq 1, b \neq 2, c \neq 3$. Put the same constraint for $(v, i)$ and $(v, j)$ as well. 
\end{itemize}

This gives a distribution over instances $I$.
We show that at least one instance satisfies the following properties:
\begin{enumerate}
    \item $I$ is not satisfied by any assignment where either $U$ or $V$ has color $4$, $5$, or $6$ appearing at least two times. 
    
    \item $I$ is not satisfied by any assignment where $U$ or $V$ has two different colors appearing at least $0.1t$ times. 
    
    \item If an assignment does not satisfy the above conditions, it means that both $u$ and $v$ have one {\em dominant} color in $[3]$ that is used at least $0.7t$ times. If $u$ and $v$ have the same dominant color, then this assignment does not satisfy $I$. 
\end{enumerate}
Given an assignment to the variables, if it does not satisfy any of the above properties, the probability that it satisfies $I$ is at most $e^{-\Omega(t^2)}$. (Actually, Properties 1. and 3. are ensured deterministically.) Since there are at most $e^{O(t)}$ assignments, for some constant $t = O(1)$, a fixed gadget $I$ with the above properties exists. 

Given a hard instance $G = (V, E)$ for $3\textsc{-Coloring}$, create an instance of 6-PAC whose vertices are $V \times [t]$ and put the above gadget for every $(u, v) \in E$. (The gadget inside $\{ u \} \times [t]$ does not depend on $v$ and can be created only once.) For any pair of variables that do not have a constraint, put the 6-PAC constraint ruling out the pairs $\{ (1, 4), (2, 5), (3, 6), (4, 1), (5, 2), (6, 3) \}$. 

If there is a valid $3$-coloring of $V$, then its natural extension to $V \times [t]$ (i.e., for $v \in V$ and $i \in [t]$, $v_i$ is assigned the color of $v$) is a valid 6-PAC assignment, and given any valid assignment for the 6-PAC instance, by choosing the dominant color for every $v \in V$ yields a valid $3$-coloring. 
\end{proof}

The hardness of \krcsp{3}{3} follows almost the same strategy, starting from \threesat on general instances. 

\threethreecsp*
\begin{proof}
We reduce from $\threesat$, whose alphabet is $\{ 0, 1 \}$. 
We want the following gadget, which is an instance of \krcsp{3}{3} with alphabet $\{ 0, 1, 2 \}$. It is parameterized by a forbidden assignment $\sigma = (\sigma_u, \sigma_v, \sigma_z) \in \{ 0, 1 \}^3$.
\begin{itemize}
\item Variable is $\{ u, v, w \} \times [t]$, for some $t = O(1)$ to be determined. Let $U := \{ u \} \times [t]$, $V := \{ v \} \times [t]$, and $W := \{ w \} \times [t]$. 

\item For any $i, j, k \in [t]$, put the constraint on $((u, i), (v, j), (w, k))$ that rules out the assignment $\sigma$. 

\item For any pair $i_1 < i_2 < i_3$, 
\begin{itemize}
    \item With probability $1/2$, sample $p \in [3], q \in \{ 0, 1 \}$ and put the constraint on $((u, i_1)$, $(u, i_2)$, $(u, i_3))$ that rules out an assignment where $(u, i_p)$ gets $q$ and the other two get $1 - q$. 
    
    \item Otherwise, put the constraint on $((u, i_1), (u, i_2), (u, i_3))$ ruling out all variables getting $2$. 
    \item Put the same constraint for $((v, i_1), (v, i_2), (v, i_3))$ and $((w, i_1), (w, i_2), (w, i_3))$. 
\end{itemize}
\end{itemize}

This gives a distribution over instances $I$.
We show that at least one instance satisfies the following properties:
\begin{enumerate}
    \item $I$ is not satisfied by any assignment where $U$, $V$, or $W$ has label $2$ appearing at least $0.1t$ times. 
    
    \item $I$ is not satisfied by any assignment where $U$, $V$, or $W$ has two different labels appearing at least $0.1t$ times. 
    
    \item If an assignment does not satisfy the above condition, it means that all $u$, $v$, and $w$ have one {\em dominant} label in $\{ 0, 1 \}$ that is used at least $0.9t$ times. If the triple of these labels is equal to $\sigma$, then the assignment does not satisfy $I$. 
\end{enumerate}
Given an assignment to the variables, if it does not satisfy the first or second property above, the probability that it satisfies $I$ is at most $e^{-\Omega(t^3)}$. The final property is ensured by construction. Since there are at most $e^{O(t)}$ assignments, for some constant $t = O(1)$, a fixed gadget $I$ with the above properties exists. 

Given a hard instance $(V, \calC)$ for \threesat, create an instance of \krcsp{3}{3} whose variables are $V \times [t]$, and 
for every $\phi \in \calC$, put the above gadget parameterized by $\sigma \in \{ 0, 1 \}^3$ not satisfying $\phi$. 
For any triple of variables that do not have a constraint, put the constraint ruling out all variables getting $2$. 

If there is a satisfying assignment for the \threesat instance, then its natural extension to $V \times [t]$ (i.e., for $v \in V$ and $i \in [t]$, $v_i$ is assigned the label of $v$) is a satisfying \krcsp{3}{3} assignment, and given any satisfying assignment for the \krcsp{3}{3} instance, choosing the dominant label for every $v \in V$ yields a valid $3$-coloring. 
\end{proof}